\documentclass[10pt,twocolumn,twoside]{IEEEtran}
\IEEEoverridecommandlockouts                              % This command is only needed if 
\usepackage{algorithmic}
\usepackage{float}
\usepackage{placeins}
\usepackage{placeins}
\usepackage{epsfig}
\usepackage{placeins}
\usepackage{booktabs}
\usepackage{array}
\usepackage{siunitx}
\usepackage{graphicx}
\usepackage{graphics}
\usepackage{amsthm}
\usepackage{textcomp}
\usepackage{url}
\usepackage{comment}
\usepackage{xspace}
\usepackage{soul}
\usepackage{xcolor}
\usepackage{cite}
\usepackage{hyperref}
\hypersetup{
    colorlinks = true,
    linkcolor = [rgb]{0,0,1},
    anchorcolor = [rgb]{0,0,1},
    citecolor = [rgb]{0.9,0.5,0},
    filecolor = [rgb]{0,0,1},
    urlcolor = [rgb]{0.9,0.5,0},
    bookmarksopen = true,
    bookmarksnumbered = true,
    breaklinks = true,
    linktocpage,
    colorlinks = true,
    linkcolor = [rgb]{0.2,0.6,0.9},
    urlcolor  = [rgb]{0,0,1},
    citecolor = [rgb]{0.9,0.5,0},
    anchorcolor = [rgb]{0.2,0.6,0.2},
}
\usepackage{amsmath,amsfonts,amssymb,amscd}
\usepackage{capt-of}

\newtheorem{theorem}{Theorem}[section]

\newtheorem{proposition}[theorem]{Proposition}

\newtheorem{remark}[theorem]{Remark}

\floatstyle{ruled}
\newfloat{model}{H}{mod}
\floatname{model}{\footnotesize Model}
\newfloat{notatio}{H}{not}
\floatname{notatio}{\footnotesize Notation}

\newcommand{\acronymA}{SA-GRU}
\newcommand{\acronymB}{$\eta$-C$\ell$-GRU}

\newenvironment{list4}{
	\begin{list}{$\bullet$}{%
			\setlength{\itemsep}{0.05cm}
			\setlength{\labelsep}{0.2cm}
			\setlength{\labelwidth}{0.3cm}
			\setlength{\parsep}{0in} 
			\setlength{\parskip}{0in}
			\setlength{\topsep}{0in} 
			\setlength{\partopsep}{0in}
			\setlength{\leftmargin}{0.16in}}}
	{\end{list}}

\makeatletter
\@ifundefined{acronym}{\newcommand{\acronym}{\textnormal{L-GRU}}}
{\renewcommand{\acronym}{\textnormal{L-GRU}}}
\@ifundefined{acronymA}{\newcommand{\acronymA}{SA-GRU}}{\renewcommand{\acronymA}{SA-GRU}}
\@ifundefined{acronymB}{\newcommand{\acronymB}{DCL-GRU}}{\renewcommand{\acronymB}{DCL-GRU}}
\makeatother

\title{\huge \bf Lightweight Gated Recurrent Unit Variants \\ for Real-Time Channel Prediction}

\author{Kyriakos Christodoulides, Kyriakos M.\ Deliparaschos, Risto Wichman, and Themistoklis Charalambous%
\thanks{K.~Christodoulides and T.~Charalambous are with the Department of Electrical and Computer Engineering, University of Cyprus, 1678 Nicosia, Cyprus. T.~Charalambous is also a Visiting Professor with the Department of Electrical Engineering and Automation, School of Electrical Engineering, Aalto University, Espoo, Finland. E-mails:{\tt~christodoulides.kyriakos@ucy.ac.cy},
{\tt~charalambous.themistoklis@ucy.ac.cy}}

\thanks{K.~M.\ Deliparaschos is with the Department of Electrical Engineering, Computer Engineering and Informatics, Cyprus University of Technology, Limassol, Cyprus. E-mail:{\tt~k.deliparaschos@cut.ac.cy}.}
\thanks{R.~Wichman is with the Department of Information and Communications Engineering, School of Electrical Engineering, Aalto University, Espoo, Finland. E-mail:{\tt~risto.wichman@aalto.fi}.}
\thanks{Preliminary results of this paper have been accepted for publication in IEEE CAMAD~\cite{christodoulides2026lightweight}. The journal version substantially extends the conference paper through a new variant of lightweight GRU, herein called Doubly Constrained Lightweight GRU (DCL-GRU), a more detailed system and channel formulation, an expanded methodological and theoretical treatment, and a broader numerical evaluation across model architectures, optimisation budgets, computational cost, and received SNR.}
}

\begin{document}

\maketitle

% =================================
% Abstract
% =================================
\begin{abstract}
Machine-learning-based channel predictors must operate under stringent latency, memory, and computational constraints while remaining robust to noisy and time-varying observations. This paper develops a causal channel-prediction framework based on three single-layer gated recurrent unit variants: an unconstrained lightweight GRU (L-GRU), a stability-aware GRU (SA-GRU) with a spectral bound on the candidate-state recurrent matrix, and a doubly constrained lightweight GRU (DCL-GRU) with additional control of the reset-gate recurrent matrix. A sufficient condition is derived for contraction of the complete candidate-state mapping, while preserving the parameter count and inference-time structure of the baseline architecture. These guarantees apply to the candidate-state
mapping and do not directly imply contraction of the complete GRU hidden-state transition. The models are trained on $2\times2$ MIMO channels generated using the 3GPP CDL-A model, and their hyperparameters are selected through Bayesian optimisation with Optuna's Tree-structured Parzen Estimator. Across the considered SNR range, the constrained variants retain competitive prediction accuracy and achieve optimisation runtimes close to L-GRU, with speedups of $1.72\times$ and $1.76\times$ relative to a five-layer GRU for SA-GRU and DCL-GRU, respectively. All audited constrained runs satisfy the prescribed spectral bounds. Under temporary observation corruption followed by recursive prediction, SA-GRU reduces the mean and peak hidden-state deviations by approximately $15.3\%$ and $13.0\%$, respectively, relative to L-GRU, whereas L-GRU achieves the lowest rollout NMSE. These results highlight an explicit trade-off between
prediction accuracy, empirical rollout robustness, and candidate-state stability guarantees.
\end{abstract}


% =================================
% Introduction
% =================================
\section{Introduction}
\label{sec:intro}

Wireless channels are inherently time-varying and frequency-selective due to multipath propagation, user mobility, and environmental dynamics~\cite{hlawatsch2011wireless}. The strong temporal correlation induced by Doppler shifts and multipath fading makes channel estimation and prediction an inherently sequential task, requiring models that can capture time-evolving dependencies across successive channel realisations. Accurate channel estimation and prediction are crucial for coherent demodulation, adaptive modulation, and resource allocation in modern communication systems, particularly in rapidly time-varying or high-mobility scenarios.

% ---------------------------------------------------
% Model-based channel estimation and prediction
% ---------------------------------------------------
Time-varying channel estimation and prediction have traditionally relied on model-based approaches grounded in analytical Doppler models. In particular, physical-model-based estimators have been proposed in~\cite{Lindbom,Bosisio}, where autoregressive (AR) and ARMA processes are fitted to the classical Jakes Doppler spectrum~\cite{Jakes} and embedded into a Kalman filtering framework for channel tracking. Similar methodologies have been extended to long-range channel prediction in~\cite{Ekman_Thesis,Ekman1,Sternad2}, exploiting the fact that the signal bandwidth is typically much larger than the maximum Doppler frequency. Despite their solid theoretical foundation, these approaches are inherently sensitive to model mismatch: practical propagation environments often deviate from the idealised Jakes spectrum, leading to an irreducible error floor in prediction performance~\cite{Zhao}. A comprehensive review in~\cite{Hallen_Dec07} shows that AR-based predictors generally outperform sum-of-sinusoids (SoS) methods for short prediction horizons under statistical channel models, whereas SoS predictors achieve longer horizons mainly in synthetic scenarios.

% ---------------------------------------------------
% Basis Expansion Models
% ---------------------------------------------------
To mitigate sensitivity to Doppler model assumptions, Basis Expansion Model (BEM) estimators represent the time-varying channel as a low-dimensional expansion over deterministic basis functions, making them more robust to spectral mismatch~\cite{Niedzwiecki_book,Tsatsanis}. BEM-based estimators have been shown to outperform Doppler-model-based Kalman filters in terms of MSE under mismatch conditions. Several basis choices have been investigated, including discrete prolate spheroidal sequences (DPSS)~\cite{Zemen1} and discrete cosine transform (DCT) bases~\cite{Schmidt11dct,Tang_May07}. Recursive BEM formulations can be embedded into a Kalman filtering framework with time-invariant dynamics, yielding a stationary Kalman gain and significantly reduced computational complexity~\cite{Lindbom,Bosisio,Muralidhar_Jul09}. While such approaches achieve improved robustness and longer prediction horizons compared to AR-based predictors, they remain constrained by the choice of fixed basis functions.

% ---------------------------------------------------
% Limitations of model-based approaches
% ---------------------------------------------------
Despite their efficiency and interpretability, model-based and BEM-based predictors depend on modeling assumptions, limiting their adaptability in nonlinear, nonstationary, or heterogeneous propagation environments. As wireless systems evolve toward highly dynamic scenarios with varying mobility, traffic patterns, and heterogeneous radio propagation environments, these assumptions become increasingly restrictive.

% ---------------------------------------------------
% Deep Neural Sequence Models
% ---------------------------------------------------
Recent advances in machine learning have led to the emergence of \emph{Deep Neural Sequence Models} (DNSMs) for temporal prediction tasks. DNSMs encompass recurrent architectures such as Recurrent Neural Networks (RNNs)~\cite{info15090517}, including GRU and LSTM variants, as well as causal convolutional neural networks (CNNs)~\cite{CNNsurvey:2022}. RNNs offer a particularly attractive data-driven alternative for channel prediction by learning temporal channel dynamics directly from observations. By maintaining and recursively updating an internal hidden state, RNN-based architectures can capture temporal correlations and adapt to evolving channel statistics without relying on explicit Doppler models or predefined basis expansions. This makes them well suited for nonstationary wireless environments. In contrast, CNN-based approaches primarily exploit spatial or short-range temporal correlations and typically require stacked historical inputs to process sequential data. Such input buffering increases latency and computational overhead, limiting their suitability for real-time PHY-layer deployment where strict causality and low latency are essential.

% ---------------------------------------------------
% GRUs and lightweight recurrent architectures
% ---------------------------------------------------
Among recurrent architectures, Gated Recurrent Units (GRUs)~\cite{DBLP:journals/corr/ChoMBB14} provide a favourable balance between modelling capacity and computational efficiency, making them particularly well suited for causal PHY-layer channel prediction~\cite{Salem2022}. GRUs maintain a hidden state that evolves recursively over time, allowing them to capture temporal correlations in the channel coefficients and to adapt to time-varying propagation conditions. Compared to vanilla RNNs, the gating mechanism in GRUs mitigates vanishing and exploding gradient issues, enabling stable learning of both short-term fluctuations and longer-term channel trends. In contrast to Long Short-Term Memory (LSTM) networks~\cite{SHERSTINSKY:2020}, GRUs employ a simplified architecture in which the input and forget gates are combined into a single update gate, while a reset gate controls the incorporation of new information. This reduction in architectural complexity leads to a significantly lower number of trainable parameters and reduced inference latency, while preserving comparable expressive power for modelling temporal dynamics. These properties make GRUs particularly attractive for real-time wireless systems, where tight latency, energy, and hardware constraints limit the feasibility of deeper or more complex recurrent models.

Despite these advantages, deploying standard GRU architectures at the physical layer remains challenging in highly resource-constrained or latency-sensitive scenarios, such as edge devices, IoT nodes, or FPGA-based implementations. In such settings, even moderate hidden-state dimensions may result in non-negligible computational overhead and memory usage, especially when operating at high symbol rates or under stringent real-time requirements. This motivates the design of \emph{lightweight recurrent architectures} that preserve the temporal modelling capability of GRUs while further reducing computational and implementation complexity.

Furthermore, most existing approaches primarily focus on improving prediction accuracy through architectural complexity, optimiser selection, or online adaptation mechanisms, while implicitly assuming that stability emerges from training. In practical wireless systems, however, this assumption is often violated. Channel predictors are deployed in a strictly causal and streaming fashion, operate open-loop for extended horizons, and are subject to non-stationary dynamics induced by mobility, Doppler variations, and abrupt environmental changes. In such regimes, even small modelling errors or distribution shifts may accumulate over time, leading to unstable hidden-state trajectories, sensitivity to initialisation, and degraded long-term performance. These effects are particularly evident when recurrent models are repeatedly queried without frequent retraining or reset, as is common in low-latency PHY-layer operation.

Recent work such as ERSO-GRU~\cite{liu_wireless_2023} addresses part of this challenge by incorporating experience replay and optimiser-driven hyperparameter tuning to improve online adaptation and mitigate error accumulation. While effective in practice, this approach treats stability as an emergent property of data-driven adaptation rather than as a structural property of the recurrent dynamics themselves. As a result, stability is not guaranteed \emph{a priori} and remains sensitive to training conditions, replay buffer composition, and optimiser behaviour. In contrast, stability-by-design learning architectures explicitly embed dynamical stability constraints into the model structure. This paradigm has gained increasing attention in related domains, including contractive recurrent neural networks~\cite{NEURIPS2021_Jafarpour,9867357,10582451} and Lyapunov-stable policies in reinforcement learning~\cite{certifyingstabilityreinforcementlearning}. The central idea is to design the model such that its internal state evolution satisfies provable stability properties independently of the specific training data realisation.

Motivated by this perspective, we investigate spectral control of selected recurrent pathways in lightweight GRU predictors. We introduce SA-GRU, which bounds the direct recurrent gain entering the candidate state, and DCL-GRU, which additionally constrains the reset-gate recurrent matrix and satisfies a sufficient condition for contraction of the complete candidate-state mapping. These guarantees do not imply contraction of the complete GRU hidden-state transition, which also depends on the update gate and the direct memory pathway. The proposed constraints preserve the inference-time architecture and parameter count of L-GRU while providing explicit control of recurrent gains. Numerical results under noisy $2\times2$ MIMO channels evaluate the resulting trade-offs among one-step prediction accuracy, computational cost, structural guarantees, and empirical perturbation sensitivity.

The main contributions of this work are summarised as follows:
\begin{list4}
    %\item We formulate a causal channel-prediction pipeline based on single-layer lightweight GRU variants for noisy $2\times2$ MIMO CDL-A channels.
        \item We formulate a causal channel-prediction pipeline based on single-layer lightweight GRU variants for noisy MIMO channels.
    \item We introduce SA-GRU and DCL-GRU and derive spectral conditions that bound the conditional candidate-state pathway and guarantee contraction of the complete candidate-state mapping, respectively.
    \item We evaluate the proposed models through Bayesian hyperparameter optimisation across received SNR, optimisation budget, runtime, and spectral-constraint compliance.
    \item We conduct a paired corrupted-rollout experiment that quantifies the trade-off between recursive prediction accuracy and hidden-state perturbation sensitivity.
\end{list4}

The remainder of this paper is organised as follows. Section~\ref{sec:prelim} introduces the notation and preliminaries. Section~\ref{sec:modelling} presents the MIMO channel model and data representation. Section~\ref{sec:preprossesing} describes the preprocessing steps. Section~\ref{sec:lightweightGRU} introduces the L-GRU, SA-GRU, and DCL-GRU architectures. Section~\ref{sec:bayesian} presents the Bayesian hyperparameter optimisation framework. Section~\ref{sec:prob.sol} describes the complete prediction pipeline. Section~\ref{sec:numerical_evaluations} provides the numerical evaluations. Section~\ref{sec:conclusions} concludes the paper and discusses future directions.

% =================================
% Notation and Preliminaries
% =================================
\section{Notation and Preliminaries}
\label{sec:prelim}

\subsection{Mathematical Notation}

Scalars are denoted by lowercase letters (e.g., $x$), vectors by bold lowercase letters (e.g., $\mathbf{x}$), and matrices or higher-order tensors by bold uppercase letters (e.g., $\mathbf{X}$ and $\mathbf{H}$). Calligraphic letters (e.g., $\mathcal{L}$) denote loss functions. The operator $(\cdot)^\top$ denotes the transpose. The Euclidean norm of a vector is denoted by $\|\cdot\|_2$, and $\odot$ denotes element-wise Hadamard multiplication.

The real and imaginary parts of a complex-valued quantity $x \in \mathbb{C}$ are denoted by $\Re\{x\}$ and $\Im\{x\}$, respectively. The expectation operator is denoted by $\mathbb{E}[\cdot]$. The logistic sigmoid function is written as $\sigma(\cdot)$, and $\tanh(\cdot)$ denotes the hyperbolic tangent.

\subsection{Time and Indexing Conventions}

Discrete time is indexed by $t \in \mathbb{Z}_{\ge 0}$ and corresponds to successive channel snapshots (e.g., OFDM symbols or uniformly sampled channel observations). Delay taps are indexed by $\ell \in \{0,1,\ldots,L-1\}$, where $L$ denotes the total number of delay bins after discretisation. Subcarriers are indexed by $k \in \{0,1,\ldots,N_{\mathrm{sc}}-1\}$.

Transmit and receive antenna indices are denoted by $n_t \in \{1,\ldots,N_{\mathrm{tx}}\}$ and $n_r \in \{1,\ldots,N_{\mathrm{rx}}\}$, respectively. Independent channel realisations generated by the CDL simulator are indexed by $b \in \{1,\ldots,B\}$, where $B$ also denotes the batch size during training.

\subsection{Channel Representation}

The wireless channel is modelled as a time-varying, frequency-selective MIMO channel. Its continuous-time baseband impulse response between transmit antenna $n_t$ and receive antenna $n_r$ is denoted by $h_{n_r,n_t}(t,\tau)$, where $t$ represents time and $\tau$ denotes propagation delay. For learning and simulation purposes, this impulse response is discretised along the delay axis into $L$ delay taps and sampled uniformly in time, yielding a discrete-time, discrete-delay representation. Unless stated otherwise, the channel coefficients are assumed to be complex-valued and zero-mean, with temporal correlation induced by Doppler effects and user mobility. The CDL model introduced in Section~\ref{sec:modelling} provides a physically grounded mechanism for generating such correlated channel realisations.

% =================================
% CDL Channel Modelling and Data Representation
% =================================
\section{CDL Channel Modelling and Data Representation}
\label{sec:modelling}

A transmitted signal reaches the receiver via multiple propagation paths, each experiencing different delays, attenuation, phase shifts, and Doppler shifts. 
Accurately modelling these effects is essential for simulation, channel estimation, and predictive algorithms. 
Here, we model the MIMO channel using the 3GPP Clustered Delay Line (CDL) model~\cite{3gpp38901}.
%The 3GPP TR~38.901 specification defines the Tapped Delay Line (TDL) and Clustered Delay Line (CDL) models. 
%The former is a statistical model in which the channel taps follow a prescribed statistical distribution, while the latter is a geometric model that organises multipath components into physically meaningful clusters, providing a tractable discrete representation for simulation and learning-based modelling. 
%In the CDL framework, a \emph{cluster} represents a group of rays sharing similar propagation characteristics, such as delay and angular distribution. Each cluster contains multiple rays that may differ slightly in amplitude, phase, or Doppler shift.
%\subsection{Clusters and Channel Taps}
%\label{subsec:clusters}
The baseband-equivalent impulse response between transmit antenna $n_t$ and receive antenna $n_r$ is expressed as
\begin{equation}
	h_{n_r,n_t}(t,\tau) = \sum_{k=1}^{K} \sum_{m=1}^{M_k}
	\alpha_{k,m} e^{j 2\pi \nu_{k,m} t} \, \delta(\tau - \tau_k),
\end{equation}
where $K$ is the number of clusters, $M_k$ is the number of rays in cluster $k$, $\tau_k$ denotes the mean delay of cluster $k$, $\alpha_{k,m}$ is the complex gain of ray $m$ in cluster $k$, and $\nu_{k,m}$ is the associated Doppler shift. 
%Clusters thus represent the physical origin of multipath propagation, while individual rays within each cluster provide fine-grained multipath diversity.

To simulate channels in discrete time, the continuous impulse response is sampled along the delay axis into a finite number of \emph{channel taps}, each representing the aggregate contribution of all rays whose delays fall within the corresponding delay bin $\tau_\ell = \ell T_s$,~ $\ell = 0, 1, \ldots, L-1$, where $T_s$ is the sampling interval and $L$ is the total number of taps. The discrete tap values $h(t, \tau_\ell)$ are complex-valued and vary over time due to Doppler shifts. Each cluster may span multiple taps depending on the discretisation resolution, and each tap may contain contributions from multiple clusters. Conceptually, clusters describe the physical structure of the multipath environment, while taps provide a computationally tractable representation suitable for simulation and neural network input.

\subsection{Discrete CDL Channel Tensor}

The CDL simulator generates a discrete-time channel tensor encompassing all transmit and receive antennas, time samples, and delay taps. For a batch of $B$ independent channel realisations, the tensor is represented as $\mathbf{H} \in \mathbb{C}^{[B,\, N_{\text{rx}},\, N_{\text{tx}},\, T,\, L]}$, where $N_{\text{rx}}$ and $N_{\text{tx}}$ are the numbers of receive and transmit antennas, $T$ is the number of time samples (e.g., OFDM symbols or snapshots).
%, and $L$ is the number of channel taps (as defined in Section~\ref{subsec:clusters}). 
The last two dimensions $(T, L)$ describe the temporal evolution of the channel impulse response over time and delay, while the middle dimensions $(N_{\text{rx}}, N_{\text{tx}})$ correspond to the spatial domain across the MIMO antenna elements. The first dimension $B$ indexes independent channel realisations generated under different user positions, Doppler conditions, or random seeds. The tensor $\mathbf{H}$ thus compactly encapsulates the spatiotemporal structure of a wideband, time-varying MIMO channel. Each complex entry
\begin{equation}
	H[b, n_r, n_t, t, \ell] \approx h_{n_r,n_t}(t, \tau_\ell)
\end{equation}
represents the small-scale fading gain of the multipath component arriving at delay $\tau_\ell$ during the $t$-th time instant for the $b$-th channel realisation.

In OFDM systems, data symbols are transmitted simultaneously over multiple orthogonal subcarriers in the frequency domain. 
%Since the CDL simulator produces the time-domain impulse response, it is necessary to transform this representation into the frequency domain to obtain per-subcarrier channel coefficients. This is achieved by 
Applying a discrete Fourier transform (DFT) along the delay dimension:
\begin{equation}
	H_f[b, n_r, n_t, t, k] =
	\sum_{\ell=0}^{L-1}
	H[b, n_r, n_t, t, \ell]
	\,e^{-j 2\pi k \ell / N_{\text{FFT}}},
\end{equation}
where $k \in \{0, 1, \ldots, N_{\text{FFT}}-1\}$ denotes the subcarrier index and $N_{\text{FFT}}$ is the number of FFT points. The resulting frequency-domain tensor
\begin{equation}
	\mathbf{H}_f \in \mathbb{C}^{[B,\, N_{\text{rx}},\, N_{\text{tx}},\, T,\, N_{\text{sc}}]}
\end{equation}
represents the complex channel gain for each subcarrier across time, antennas, and independent realisations, where $N_{\text{sc}} \leq N_{\text{FFT}}$ denotes the number of active subcarriers. As the user or scatterers move, the Doppler effect introduces time variations in $\mathbf{H}_f$, leading to temporal correlation across consecutive OFDM symbols.

\begin{remark}
	In practice, only a subset of subcarriers (pilot subcarriers) are known \emph{a priori} at the receiver. These pilots provide noisy estimates $\tilde{H}_f[b, n_r, n_t, t, k]$ of the true channel coefficients, which serve as partial observations of the underlying channel dynamics. Data-driven recurrent models such as the proposed GRU variants can then be trained to exploit the temporal correlation in these pilot-based sequences to predict future channel states or reconstruct unobserved subcarrier coefficients. This frequency-domain representation forms the direct input to neural architectures designed for channel estimation and prediction in OFDM-based communication systems.
\end{remark}

% =================================
% Preprocessing
% =================================
\section{Preprocessing CDL Outputs for GRU Prediction}
\label{sec:preprossesing}

Before feeding the CDL tensor into the GRU, several preprocessing steps are required.

\subsection{Real-Imaginary Separation}

The complex channel coefficients are first separated into their real and imaginary components. Given the frequency-domain channel tensor $\mathbf{H}_f \in \mathbb{C}^{[B,\, N_{\text{rx}},\, N_{\text{tx}},\, T,\, N_{\text{sc}}]}$, the real and imaginary parts are stacked along a new feature dimension to produce a real-valued tensor $\mathbf{X} \in \mathbb{R}^{[B,\, T,\, N_{\text{rx}} \times N_{\text{tx}} \times N_{\text{sc}} \times 2]}$, preserving the full information content of the complex channel coefficients.

\begin{remark}
	When pilot-based channel estimation is considered, the input sequence may include only the subcarriers corresponding to pilot tones. In such cases, the observed pilot coefficients are extracted from $\mathbf{H}_f$ before real/imaginary separation and normalisation. The GRU then predicts the full channel or future pilot subcarriers, which can be further interpolated across subcarriers if needed.
\end{remark}

\subsection{Normalisation}
\label{subsec:normalisation}

Prior to processing by the recurrent predictor, the input tensor is standardised using training-set statistics:
\begin{equation}
\mathbf{X}_{\mathrm{norm}}
=
\frac{\mathbf{X}-\boldsymbol{\mu}}{\boldsymbol{\sigma}},
\end{equation}
where $\boldsymbol{\mu}$ and $\boldsymbol{\sigma}$ are computed exclusively from the training set and reused without modification during validation and inference.

Normalisation improves gradient conditioning, reduces scale imbalance among channel features, and ensures that training and inference operate under consistent input statistics. The spectral conditions developed in Section~V concern selected recurrent pathways of the GRU rather than the complete input-to-state transition. Consequently, normalisation should be interpreted as a numerical-conditioning and data-scaling operation, not as a mechanism that independently preserves or guarantees contraction of the complete hidden-state dynamics.

From a practical perspective, standardisation prevents features associated with stronger channel coefficients from dominating the optimisation objective and produces more uniform gradient magnitudes across antenna links and temporal samples.

\subsection{Batching}

After normalisation, the preprocessed channel sequences are divided into smaller temporal segments to form mini-batches suitable for training. Batching serves two purposes: it allows the recurrent predictor to learn from short, overlapping sequences of past channel states, and it enables efficient parallel computation on GPUs.

Given a long time series of preprocessed channel features $\{\mathbf{X}_1, \mathbf{X}_2, \ldots, \mathbf{X}_{T_{\text{total}}}\}$, a sliding window of fixed length $T_{\text{seq}}$ is applied to extract shorter subsequences for prediction. Each subsequence of $T_{\text{seq}}$ consecutive time steps serves as one training input, and the immediate next sample forms the corresponding target:
\begin{equation}
	\big\{ (\mathbf{X}_{t:t+T_{\text{seq}}-1},\; \mathbf{X}_{t+T_{\text{seq}}}) \big\}_{t=1}^{T_{\text{total}}-T_{\text{seq}}}.
\end{equation}
The resulting input sequences are grouped into batches of size $B$, producing the training tensor
\begin{equation}
	\mathbf{X}_{\text{batch}} \in \mathbb{R}^{[B,\, T_{\text{seq}},\, F]},
\end{equation}
where $F = N_{\text{rx}} \times N_{\text{tx}} \times N_{\text{sc}} \times 2$ denotes the feature dimension. Each batch contains $B$ temporal windows, enabling efficient gradient updates and data-parallel training.

\subsection{Regularisation}
\label{subsec:Regularisation}

Although the proposed recurrent models use a single recurrent layer, the hidden-state dimension selected through Bayesian optimisation may still provide substantial modelling capacity. This creates a risk of overfitting when the models are trained on strongly correlated CDL-generated channel sequences. Regularisation is therefore used to prevent the predictor from memorising individual channel trajectories and to encourage generalisation to unseen channel realisations.

The dropout probability $p$ controls the strength of regularisation by stochastically masking a subset of the non-recurrent activations during training. In recurrent architectures, applying dropout requires particular care, as indiscriminate masking of recurrent connections can disrupt the temporal continuity that these models are designed to capture. Dropout is therefore applied exclusively to the \emph{non-recurrent} affine transformations, i.e., the input-to-hidden projections (and inter-layer projections in multi-layer configurations). This preserves the integrity of the recurrent pathway while still providing effective regularisation, forcing the network to develop more robust, distributed representations that do not depend on isolated input features or specific structural patterns present in the CDL-generated channel sequences.

Specifically, before the input at time step $t$ is processed by the GRU gates, it is passed through a dropout layer:
\begin{equation}\label{eq:batchGRU}
	\tilde{\mathbf{X}}_{\text{batch}}  = \mathbf{M}_{\text{batch}} \odot \mathbf{X}_{\text{batch}},
\end{equation}
where
\begin{equation}
	\mathbf{M}_{\text{batch}} \in \left\{0,\;\frac{1}{1-p}\right\}^{B \times T_{\mathrm{seq}} \times F}
\end{equation}
is a binary mask whose entries are sampled i.i.d.\ from
\begin{equation}
	m^{(i)}_{b,t} =
	\begin{cases}
		0,              & \text{with probability } p,   \\[4pt]
		\dfrac{1}{1-p}, & \text{with probability } 1-p,
	\end{cases}
\end{equation}
for every batch index $b \in \{1,\ldots,B\}$, time index $t \in \{1,\ldots,T_{\mathrm{seq}}\}$, and feature index $i \in \{1,\ldots,F\}$. This operation randomly masks individual input features but does not remove entire samples or time steps, ensuring that the temporal continuity of each sequence is preserved. The rescaling by $1/(1-p)$ maintains the expected value of the activations and stabilises training.

\begin{remark}
	The optimal value of $p$ depends on several interacting factors, including $H$, the temporal variability of the channel (e.g., Doppler spread), and the diversity of the generated channel realisations, making it unsuitable to fix by heuristic rules. For this reason, $p$ is incorporated into the Bayesian hyperparameter optimisation loop (see Section~\ref{sec:bayesian}), allowing the model to automatically identify the level of regularisation that best balances capacity and generalisation across a wide range of channel conditions.
\end{remark}

% =================================
% Proposed GRU Architecture
% =================================
\section{Proposed GRU Architecture for Causal Channel Prediction}
\label{sec:lightweightGRU}

After preprocessing the CDL channel outputs into sequential, real-valued input tensors, a lightweight GRU is employed to model the temporal evolution of the wireless channel and perform causal prediction. Each training sample consists of a window of $T_\text{seq}$ consecutive time snapshots of the preprocessed channel, used to predict the channel at the next time step. Before entering the recurrent cell, the input features undergo dropout-based regularisation as described in Section~\ref{subsec:Regularisation}, yielding the masked input vectors $\tilde{\mathbf{X}}_t \in \mathbb{R}^{F}$.

\subsection{Typical Lightweight GRU Architecture}

A lightweight GRU maintains a hidden state $\mathbf{h}_t \in \mathbb{R}^{H}$ that evolves with the input sequence and encodes the temporal context of the channel. The evolution of the hidden state is governed by the standard GRU update equations:
\begin{subequations}\label{eq:lGRU}
	\begin{align}
		\mathbf{z}_t         &= \sigma\!\left(\mathbf{W}_z \tilde{\mathbf{X}}_t
		+ \mathbf{U}_z \mathbf{h}_{t-1} + \mathbf{b}_z\right), \\
		\mathbf{r}_t         &= \sigma\!\left(\mathbf{W}_r \tilde{\mathbf{X}}_t
		+ \mathbf{U}_r \mathbf{h}_{t-1} + \mathbf{b}_r\right),
        \label{eq:rt} \\
		\tilde{\mathbf{h}}_t &= \tanh\!\left(\mathbf{W}_h \tilde{\mathbf{X}}_t
		+ \mathbf{U}_h \left(\mathbf{r}_t \odot \mathbf{h}_{t-1}\right)
		+ \mathbf{b}_h\right),
        \label{eq:candidate_state_contractivity} \\
		\mathbf{h}_t         &= \left(\mathbf{1} - \mathbf{z}_t\right) \odot \mathbf{h}_{t-1}
		+ \mathbf{z}_t \odot \tilde{\mathbf{h}}_t,
	\end{align}
\end{subequations}
where $\mathbf{z}_t$ and $\mathbf{r}_t$ are the update and reset gates, respectively, $\tilde{\mathbf{h}}_t$ is the candidate hidden state, $\sigma(\cdot)$ denotes the logistic sigmoid, and $\mathbf{W}_\varkappa$, $\mathbf{U}_\varkappa$, $\mathbf{b}_\varkappa$ for $\varkappa \in \{z,r,h\}$ are trainable parameters.

\begin{remark}
The model defined by the GRU update equations in (11) and the readout in (13) corresponds to the unconstrained lightweight GRU, denoted L-GRU. SA-GRU extends this architecture by imposing the spectral bound in (14) on $\mathbf{U}_h$, whereas DCL-GRU additionally constrains $\mathbf{U}_r$ through (16)--(20).
\end{remark}

\begin{remark}
The term lightweight refers primarily to the single-layer recurrent structure, the absence of a separate LSTM-style memory cell, and the compact channel-feature representation. The hidden dimension is selected through Bayesian optimisation and is not restricted to the interval $[16,64]$. For the GRU update equations in (11) and the linear readout in (13), the parameter count is
\begin{equation}
N_{\mathrm{par}}
=
3HF+3H^2+3H+FH+F.
\end{equation}
For $F=4$ and $H=64$, this gives $13{,}508$ parameters, whereas the $H=240$ model used in the corrupted-rollout experiment contains $177{,}364$ parameters. Despite the larger hidden dimension in the latter experiment, the architecture remains single-layer and retains a substantially lower recurrent depth than the GRU-5L baseline.
\end{remark}

The proposed \acronym\ operates in a strictly \emph{causal} manner: at each time instant $t$, the hidden state $\mathbf{h}_t$ depends exclusively on the current input and on past observations, with no access to future information. After updating the hidden state, a one-step-ahead channel prediction is obtained through a linear readout:
\begin{equation}
	\hat{\mathbf{X}}_{t+1} = \mathbf{W}_o \mathbf{h}_t + \mathbf{b}_o,
\end{equation}
where $\mathbf{W}_o \in \mathbb{R}^{F \times H}$ and $\mathbf{b}_o \in \mathbb{R}^{F}$ are the trainable output weight matrix and bias vector.

% ---------------------------------------------------
% SA-GRU
% ---------------------------------------------------
\subsection{Stability-Aware GRU with Spectral Projection (\acronymA)}
\label{subsec:sa_gru}

While the lightweight GRU in~\eqref{eq:lGRU} provides an efficient architecture for real-time channel prediction, its stability properties are not explicitly controlled. In wireless applications, the predictor operates recursively over long horizons under potentially non-stationary conditions, where uncontrolled recurrent dynamics may amplify perturbations or accumulate numerical errors.

To address this, we introduce \acronymA, a stability-aware variant in which a spectral constraint is imposed on the recurrent matrix $\mathbf{U}_h$ associated with the candidate hidden state. Specifically, we require
\begin{equation}
\left\|\mathbf{U}_h\right\|_2 \le \rho_h,
\label{eq:spectral_constraint_final}
\end{equation}
where $\|\cdot\|_2$ denotes the spectral norm and $\rho_h\in(0,1)$ is a user-defined margin controlling the maximum admissible gain of the candidate-state recurrent operator. All remaining parameters and gating equations in~\eqref{eq:lGRU} are left unchanged. This constraint directly limits the amplification introduced by $\mathbf{U}_h$ along the recurrent pathway entering the candidate hidden state. In particular, for a fixed realisation of the reset gate $\mathbf{r}_t$, the mapping from $\mathbf{h}_{t-1}$ to the candidate state has a Lipschitz constant no greater than $\|\mathbf{U}_h\|_2$, since $\tanh(\cdot)$ is globally $1$-Lipschitz and the entries of $\mathbf{r}_t$ lie in $(0,1)$. Thus,~\eqref{eq:spectral_constraint_final} guarantees contraction of this conditional candidate-state pathway with margin $1-\rho_h$.

To enforce~\eqref{eq:spectral_constraint_final}, we apply a one-time post-training spectral projection to $\mathbf{U}_h$. Let $\widehat{\mathbf{U}}_h$ denote the unconstrained matrix obtained through standard stochastic gradient-based optimisation. After convergence, the deployed recurrent matrix is defined as
\begin{equation}
\mathbf{U}_h
=
\frac{\rho_h}
{\max\!\left(
\left\|\widehat{\mathbf{U}}_h\right\|_2,\,
\rho_h
\right)}
\widehat{\mathbf{U}}_h.
\label{eq:uh_rescaling}
\end{equation}
This projection leaves $\widehat{\mathbf{U}}_h$ unchanged whenever it already satisfies~\eqref{eq:spectral_constraint_final}; otherwise, it rescales the matrix so that $\|\mathbf{U}_h\|_2 = \rho_h$.

The training procedure remains unchanged: all parameters are first learned using a standard optimiser (e.g., Adam~\cite{adam}), and the spectral projection is applied only after convergence. The proposed modification therefore does not alter the backpropagation procedure, introduce additional trainable parameters, or add gates, hidden states, or matrix multiplications during inference. Its contribution is instead to provide a simple, deployment-stage mechanism for explicitly controlling the spectral gain of the candidate-state recurrent operator while preserving the architecture and per-step computational structure of the original lightweight GRU. Accordingly, \acronymA\ is referred to as \emph{stability-aware} rather than globally contractive by construction.

\begin{proposition}[Spectral-gain bound of \acronymA]
\label{prop:acronym_spectral_bound}
Consider the candidate-state update of \acronymA\ in~\eqref{eq:candidate_state_contractivity}. Suppose that the post-training projection enforces $\|\mathbf{U}_h\|_2 \le \rho_h$ with $\rho_h\in(0,1)$. Then, for any fixed input $\mathbf{x}_t$ and fixed reset-gate realisation $\mathbf{r}_t$, the candidate-state mapping is Lipschitz continuous with respect to $\mathbf{h}_{t-1}$ with Lipschitz constant at most $\rho_h$. Therefore, the proposed projection bounds the direct recurrent gain of the candidate-state pathway by $\rho_h$ without changing the inference-time architecture of the lightweight GRU.
\end{proposition}
\begin{proof}
See Appendix~\ref{appendix:A}.
\end{proof}

The proposed projection introduces an explicit structural bound on one of the principal recurrent pathways. Stronger guarantees for the complete candidate-state and hidden-state mappings require additional conditions on the reset- and update-gate recurrent matrices, as discussed next.

% ---------------------------------------------------
% DCL-GRU
% ---------------------------------------------------
\subsection{Doubly Constrained Lightweight GRU (\acronymB)}
\label{subsec:contractive_lgru}

In a standard GRU, $\mathbf{r}_t$ depends on $\mathbf{h}_{t-1}$ through~\eqref{eq:rt}. Consequently, the constraint on $\mathbf{U}_h$ controls only the direct recurrent gain of the candidate-state branch, rather than providing a stand-alone guarantee of contraction of the complete mapping $\mathbf{h}_{t-1}\mapsto\widetilde{\mathbf{h}}_t$. Contraction of the complete hidden-state transition additionally depends on the update-gate dynamics, and the dependence of the reset gate on $\mathbf{h}_{t-1}$ introduces an additional term involving $\mathbf{U}_r$.

To obtain a sufficient condition for contraction of the complete candidate-state mapping, we enforce spectral-norm bounds on both $\mathbf{U}_h$ and $\mathbf{U}_r$:
\begin{align}
\left\|\mathbf{U}_h\right\|_2 &\le \rho_h, \label{eq:joint_spectral_constraints_h}\\
\left\|\mathbf{U}_r\right\|_2 &\le \rho_r, \label{eq:joint_spectral_constraints_r}
\end{align}
where $\rho_h>0$ and $\rho_r\ge 0$. These constraints are imposed after training through spectral projections. Let $\widehat{\mathbf{U}}_h$ and $\widehat{\mathbf{U}}_r$ denote the matrices obtained through unconstrained optimisation. The deployed matrices are defined as
\begin{align}
\mathbf{U}_h
&=
\frac{\rho_h}
{\max\!\left(
\left\|\widehat{\mathbf{U}}_h\right\|_2,\,
\rho_h
\right)}
\widehat{\mathbf{U}}_h,
\label{eq:Uh_projection}
\\
\mathbf{U}_r
&=
\frac{\rho_r}
{\max\!\left(
\left\|\widehat{\mathbf{U}}_r\right\|_2,\,
\rho_r
\right)}
\widehat{\mathbf{U}}_r.
\label{eq:Ur_projection}
\end{align}

\begin{proposition}[Contraction of the complete candidate-state mapping]
\label{prop:complete_candidate_contraction}
Consider the GRU candidate-state mapping~\eqref{eq:candidate_state_contractivity}, where $\mathbf{r}_t$ is given by~\eqref{eq:rt}. Assume that $\|\mathbf{h}_{t-1}\|_\infty\le 1$ and that the deployed recurrent matrices satisfy $\|\mathbf{U}_h\|_2 \le \rho_h$ and $\|\mathbf{U}_r\|_2 \le \rho_r$. If, for some $\delta \in (0,1)$,
\begin{equation}
\rho_h
\left(
1+\frac{\rho_r}{4}
\right)
\leq 1-\delta,
\label{eq:dcl_contraction_condition}
\end{equation}
then, for every fixed input $\mathbf{x}_t$, the complete candidate-state mapping $\mathbf{h}_{t-1}\mapsto\widetilde{\mathbf{h}}_t$ is contractive in the Euclidean norm with contraction factor at most $1-\delta$. The post-training projections in~\eqref{eq:Uh_projection} and~\eqref{eq:Ur_projection} enforce the required spectral-norm bounds without modifying the inference-time architecture.
\end{proposition}
\begin{proof}
See Appendix~\ref{appendix:B}.
\end{proof}
The condition in Proposition~V.4 is sufficient but is not claimed to be necessary. Consequently, candidate-state
contraction may still arise for parameter configurations that do not satisfy~(20). Characterising the tightness of the bound
and deriving less conservative contraction conditions remain important directions for future research.
\begin{remark}
\label{rem:candidate_not_full_gru}
Proposition~\ref{prop:complete_candidate_contraction} establishes contraction of the complete candidate-state mapping, including the dependence of the reset gate on the previous hidden state. It does not, by itself, establish contraction of the complete GRU hidden-state transition, because the latter additionally involves the update gate $\mathbf{z}_t$, its recurrent matrix $\mathbf{U}_z$, and the direct memory pathway through $\mathbf{h}_{t-1}$.
\end{remark}
Consequently, empirical evaluation remains necessary to assess the practical impact of the proposed constraints on recursive prediction behaviour.
\subsubsection{Computational Considerations}

DCL-GRU preserves the inference-time computational structure of L-GRU. For the main experiments, the only additional deployment-stage operations are the spectral projections of $\mathbf{U}_h$ and $\mathbf{U}_r$ in (18) and (19), respectively. These projections are performed outside the inference loop and introduce no additional trainable parameters, recurrent states, gates, or per-step matrix multiplications. Consequently, the parameter count and inference-time arithmetic complexity remain identical to those of L-GRU with the same hidden dimension.

\begin{remark}
SA-GRU and DCL-GRU provide explicit structural control of selected candidate-state recurrent pathways without modifying the inference-time architecture. SA-GRU bounds the conditional candidate-state recurrent gain, whereas DCL-GRU satisfies a sufficient condition for contraction of the complete candidate-state mapping. Neither condition alone establishes contraction of the complete GRU hidden-state transition.
\end{remark}

\subsection{Training Objective and Constraint Enforcement}
\label{subsec:training_contractivelgru}

The recurrent models are trained to predict the next normalised noisy channel observation from a window of preceding noisy observations. Let $\mathbf{y}_t$ denote the normalised observed feature vector. The training objective is
\begin{equation}
\mathcal{L}(\boldsymbol{\theta})
=
\frac{1}{T}
\sum_{t=1}^{T}
\left\|
\widehat{\mathbf{y}}_t-\mathbf{y}_t
\right\|_2^2,
\label{eq:training_objective}
\end{equation}
where $\boldsymbol{\theta}$ collects the trainable recurrent and output-layer parameters.

For the main one-step prediction and Bayesian-optimisation experiments, all recurrent parameters are first learned using unconstrained stochastic gradient-based optimisation, such as Adam~\cite{adam}. After training has converged, SA-GRU is obtained by projecting the candidate-state recurrent matrix $\mathbf{U}_h$ onto the spectral-norm constraint set defined in~\eqref{eq:uh_rescaling}. DCL-GRU additionally projects $\mathbf{U}_h$ and the reset-gate recurrent matrix $\mathbf{U}_r$ using the projections in (18) and (19), respectively. These operations ensure that SA-GRU satisfies its candidate recurrent-gain bound and that DCL-GRU satisfies the sufficient candidate-state contraction condition in (20). Because the projections are performed after unconstrained training and outside the inference loop, they introduce no additional backpropagation cost, trainable parameters, recurrent operations, or per-step inference overhead.

A different constraint-enforcement protocol is used only in the corrupted recursive-rollout experiment of Section~\ref{sec:corrupted_rollout}. In that experiment, all three variants are initialised from the same pretrained L-GRU parameters to provide a controlled comparison. The SA-GRU and DCL-GRU copies are first projected onto their respective constraint sets and are subsequently fine-tuned using a reduced learning rate. The relevant projection is reapplied after every optimiser update, allowing the remaining model parameters to adapt to the constrained recurrent matrices while continuously maintaining the prescribed spectral bounds. This experiment-specific fine-tuning procedure is not used in the main one-step prediction or Bayesian-optimisation experiments.

\subsection{Comparison to CNN and LSTM Architectures}

CNNs, LSTMs, and GRUs represent three distinct modelling paradigms with different inductive biases and computational characteristics. CNNs are effective at extracting local feature patterns but lack an intrinsic recurrent state for causal channel prediction unless multiple temporal observations are explicitly stacked at the input. LSTMs provide strong sequential modelling capacity but incur higher parameter, memory, and computational costs because of their multiple gates and separate cell state. GRUs provide an intermediate design by retaining recurrent temporal memory while using a simpler gated architecture. The proposed L-GRU, SA-GRU, and DCL-GRU variants further retain a single recurrent layer and introduce no additional inference-time gates or states, making them suitable for low-latency and hardware-constrained physical-layer prediction.

\begin{table*}[t]
\centering
\caption{Qualitative comparison of the considered channel-prediction architectures.}
\label{tab:architecture_comparison}
\scriptsize
\setlength{\tabcolsep}{5pt}
\begin{tabular}{lccccc}
\toprule
Property & CNN & LSTM & L-GRU & SA-GRU & DCL-GRU \\
\midrule
Intrinsic temporal state & No & Cell and hidden & Hidden & Hidden & Hidden \\
Streaming suitability & Moderate & Moderate & High & High & High \\
Recurrent depth used & -- & 1 & 1 & 1 & 1 \\
Explicit recurrent spectral bound & No & No & No & $\mathbf{U}_h$ & $\mathbf{U}_h,\mathbf{U}_r$ \\
Candidate-state guarantee & No & No & No & Conditional pathway bound & Complete mapping contraction \\
Complete hidden-state contraction & No & No & No & No & No \\
Inference overhead from constraint & -- & -- & -- & None & None \\
\bottomrule
\end{tabular}
\end{table*}

% =================================
% Bayesian Hyperparameter Optimisation
% =================================
\section{Bayesian Hyperparameter Optimisation}
\label{sec:bayesian}

In addition to the trainable parameters $\boldsymbol{\theta}$ learned by gradient descent, the performance of the GRU-based channel predictor depends critically on a set of hyperparameters $\boldsymbol{\phi}$ that define both the model architecture and the training procedure:
\begin{equation}
	\boldsymbol{\phi} = \{ H, \eta, p, B, T_\text{seq} \},
\end{equation}
where $H$ is the hidden-state dimension, $\eta$ is the learning rate, $p$ is the dropout probability, $B$ is the batch size, and $T_\text{seq}$ is the sequence length. Unlike $\boldsymbol{\theta}$, these hyperparameters are not optimised directly through backpropagation; instead, they are selected externally by minimising the validation NMSE. The recurrent depth is fixed by design in order to isolate the effect of the proposed lightweight and constrained recurrent cells. The lightweight recurrent variants, the LSTM baseline, and the single-layer GRU reference use one recurrent layer, while a 5-layer GRU is included as a deeper high-capacity baseline. The remaining hyperparameters, including hidden dimension, learning rate, dropout probability, batch size, and sequence length, are selected by Bayesian optimisation.

Hyperparameter optimisation is inherently computationally intensive, requiring repeated training--evaluation cycles that can take hours or days for deep architectures. Exhaustive grid or random search strategies therefore rapidly become infeasible in high-dimensional hyperparameter spaces.

\subsection{Bayesian Optimisation Framework}

Bayesian Optimisation (BO) provides a data-efficient framework for finding the optimal configuration $\boldsymbol{\phi}^*$ that minimises a validation loss $\mathcal{J}(\boldsymbol{\phi})$.\footnote{The validation loss measures how accurately the GRU predicts unseen channel sequences. A lower validation loss indicates better generalisation, meaning the model can reliably forecast future channel conditions it has not encountered during training.} The function $\mathcal{J}(\boldsymbol{\phi})$ is treated as a black-box, expensive-to-evaluate objective, and Bayesian optimisation seeks to minimise it by systematically incorporating information from past evaluations.

BO is typically realised by Sequential Model-Based Optimisation (SMBO) methods. At each iteration, a surrogate model (described in Section~\ref{subsec:surrogate}) is updated using the most recent data, and an acquisition function (described in Section~\ref{subsec:acquisitionfunction}) is optimised to balance exploration of uncertain regions with exploitation of known high-performing configurations. This process continues until the predefined trial budget is exhausted, yielding the best observed hyperparameter configuration under the available computational budget.

\subsection{Surrogate Probability Model of the Objective Function}
\label{subsec:surrogate}

Rather than treating each hyperparameter configuration independently, BO constructs a \emph{probabilistic surrogate model} that captures the relationship between hyperparameters and corresponding objective function values:
\begin{equation}
	p(\text{score} \mid \text{hyperparameters}).
\end{equation}
By explicitly modelling uncertainty, the surrogate guides the search toward regions of the hyperparameter space likely to yield improved performance, while avoiding unnecessary evaluations in less promising areas.

Let $\mathcal{D}_t = \{(\boldsymbol{\phi}_i, \mathcal{J}(\boldsymbol{\phi}_i))\}_{i=1}^t$ denote the set of $t$ evaluated configurations and their corresponding validation losses. A classical choice for the surrogate model is a \emph{Gaussian Process} (GP), which defines a distribution over functions,
$\mathcal{J}(\boldsymbol{\phi}) \sim \mathcal{GP}\!\big(m(\boldsymbol{\phi}), k(\boldsymbol{\phi}, \boldsymbol{\phi}')\big)$,
where $m(\cdot)$ is the mean function and $k(\cdot,\cdot)$ is a covariance kernel. Conditioning on $\mathcal{D}_t$, the GP yields a posterior predictive distribution for any candidate $\boldsymbol{\phi}_*$, i.e.,  
$p\big(\mathcal{J}(\boldsymbol{\phi}_*) \mid \mathcal{D}_t\big)
	= \mathcal{N}\!\big(\mu_t(\boldsymbol{\phi}_*), \sigma_t^2(\boldsymbol{\phi}_*)\big)$,
where $\mu_t(\boldsymbol{\phi}_*)$ is the predicted validation loss and $\sigma_t^2(\boldsymbol{\phi}_*)$ quantifies the associated uncertainty. This probabilistic representation enables BO to balance \emph{exploration} (sampling in regions of high uncertainty) and \emph{exploitation} (refining the search around configurations predicted to yield low validation loss).

\subsection{Acquisition Function, Surrogate Update, and Optuna TPE Implementation}
\label{subsec:acquisitionfunction}

While GP-based BO provides a convenient theoretical framework, its computational complexity scales cubically with the number of observations, which can become prohibitive as the number of trials grows. In this work, BO is implemented using the \emph{Optuna} framework~\cite{optuna_2019}, which adopts the \emph{Tree-structured Parzen Estimator} (TPE) as its default surrogate model. Instead of modelling the objective function directly, TPE models the inverse conditional density $p(\boldsymbol{\phi}|\mathcal{J})$ by constructing two density estimators: one for configurations associated with low objective values and one for the remainder. Specifically, given a threshold $\gamma$, TPE partitions the observations into $p(\boldsymbol{\phi}|\mathcal{J} < \gamma)$ and $p(\boldsymbol{\phi}|\mathcal{J} \geq \gamma)$, and selects new candidates by maximising the ratio between these densities. This formulation is closely related to maximising Expected Improvement, while enabling efficient handling of mixed, conditional, and high-dimensional hyperparameter spaces.

At each optimisation trial, Optuna proposes a new hyperparameter configuration $\boldsymbol{\phi}_{t+1}$ based on the TPE surrogate and the accumulated dataset $\mathcal{D}_t$. The GRU model is then trained with $\boldsymbol{\phi}_{t+1}$, and the resulting validation loss $\mathcal{J}(\boldsymbol{\phi}_{t+1})$ is observed. This new observation is appended to the dataset, yielding $\mathcal{D}_{t+1}$, and the surrogate densities are updated accordingly to guide subsequent trials.

\begin{remark}
A key motivation for adopting BO with a TPE surrogate is the high computational cost of each objective function evaluation: training a GRU-based channel predictor over multiple epochs, with complexity scaling with the hidden-state dimension, sequence length, and number of training samples. By concentrating subsequent evaluations in regions associated with favourable validation losses, the Optuna--TPE framework reduces the need for exhaustive enumeration of the hyperparameter space, enabling efficient tuning under strict time and resource constraints.
\end{remark}

% =================================
% Proposed Pipeline
% =================================
\section{Proposed Pipeline}
\label{sec:prob.sol}

This section outlines the proposed pipeline for causal channel prediction. The complete workflow, illustrated in Fig.~\ref{fig:bdiagram}, encompasses \emph{(i)} channel data generation using the CDL simulator, \emph{(ii)} data preprocessing, \emph{(iii)} model training with stochastic gradient optimisation, and \emph{(iv)} Bayesian hyperparameter tuning via the TPE implemented in Optuna~\cite{optuna_2019}. Each stage is designed to ensure efficient, robust, and adaptive prediction of time-varying wireless channels.

\begin{enumerate}
	\item The pipeline begins with the generation of complex-valued, time-varying channel coefficients using the 3GPP-compliant CDL model.

	\item These coefficients are transformed into frequency-domain representations and preprocessed into real-valued tensors, which are then batched for neural network input.

	\item The selected single-layer recurrent predictor operates causally and recursively to forecast the next noisy channel observation. Model training follows the objective defined in~\eqref{eq:training_objective}.

	\item Bayesian Optimisation with Optuna's TPE models two probability densities,
	      $l(\boldsymbol{\phi}) = p(\boldsymbol{\phi} | \mathcal{J}(\boldsymbol{\phi}) < \mathcal{J}^*)$
	      and $g(\boldsymbol{\phi}) = p(\boldsymbol{\phi} | \mathcal{J}(\boldsymbol{\phi}) \geq \mathcal{J}^*)$,
	      where $\mathcal{J}^*$ is a loss quantile. The next candidate configuration is selected by maximising the likelihood ratio:
	      \begin{equation}
		      \boldsymbol{\phi}_{t+1} = \arg\max_{\boldsymbol{\phi}} \frac{l(\boldsymbol{\phi})}{g(\boldsymbol{\phi})}.
	      \end{equation}
	      This approach balances exploration and exploitation while supporting efficient sequential or parallel hyperparameter evaluation.
\end{enumerate}

\begin{figure}[t]
    \centering
	\includegraphics[width=0.85\columnwidth]{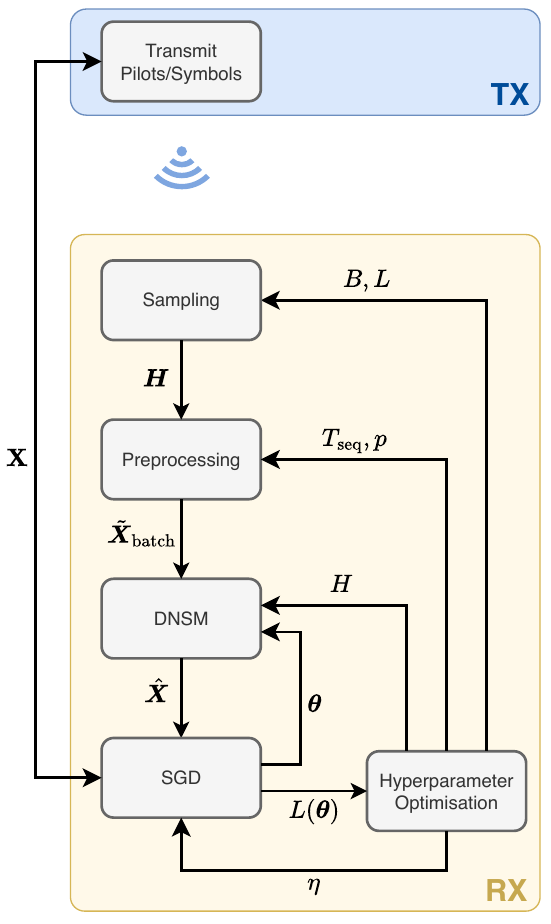}
	\caption{Block diagram of the channel prediction pipeline.}
	\label{fig:bdiagram}
\end{figure}

The integration of the \acronym\ predictor with Bayesian optimisation provides an automated, adaptive learning process that jointly enhances prediction accuracy and model efficiency. Using a single-layer \acronym\ reduces inference latency and memory consumption, making the model suitable for real-time and embedded deployments. The TPE-guided search reduces the need for exhaustive grid evaluation by allocating more trials to promising regions of the hyperparameter space. Overall, the proposed pipeline enables:
\begin{list4}
	\item \textbf{Causal, real-time prediction:} The \acronym\ updates its hidden state sequentially, allowing sample-by-sample prediction with no access to future observations.
	\item \textbf{Efficient optimisation:} Bayesian methods minimise training trials while maintaining high predictive performance.
	\item \textbf{Scalability and deployability:} The lightweight design supports implementation on GPUs, FPGAs, or embedded devices.
\end{list4}

% =================================
% Numerical Evaluations
% =================================
\section{Numerical Evaluations}
\label{sec:numerical_evaluations}

\subsection{System Model and Simulation Parameters}
\label{subsec:sim_params}

Time-varying wireless channels are generated using the 3GPP TR~38.901 CDL model. The resulting channel traces are processed into supervised learning sequences using a sliding window of length $T_{\text{seq}}$, and each recurrent model is trained to perform one-step-ahead prediction of the channel magnitude. Table~\ref{tab:sim_parameters} lists the key simulation parameters\footnote{To support reproducibility, the core project code is publicly available at \url{https://github.com/themistos/ChannelSeeksG}.}

\begin{table}[h]
\centering
\caption{Simulation parameters for the main one-step prediction and Bayesian-optimisation experiments. The corrupted recursive-rollout configuration is described separately in Section~\ref{sec:corrupted_rollout}.}
\label{tab:sim_parameters}
\begin{tabular}{ll}
\hline
\textbf{Parameter} & \textbf{Value} \\
\hline
Channel model & 3GPP TR 38.901 CDL-A \\
Carrier frequency $f_c$ & 3.5 GHz \\
Delay spread & 100 ns \\
User velocity & 30 m/s \\
Sampling frequency & 15 kHz \\
Transmission mode & Downlink \\
Antenna configuration & $2 \times 2$ MIMO, omnidirectional \\
Noise model & AWGN \\
Received SNR & $[-15,20]$ dB; $10$ dB for BO convergence \\
Total time samples $T_{\text{total}}$ & 100 \\
Sequence length $T_{\text{seq}}$ & Optimised (BO) \\
Batch size $B$ & Optimised (BO) \\
Recurrent layers & 1~$\forall$ models except for GRU-5L \\
Hidden dimension $H$ & Optimised (BO) \\
Dropout probability $p$ & Optimised (BO) \\
Optimiser & Adam \\
Training/validation obj. & Normalised Mean Squared Error (NMSE) \\
Hyperparameter search & Optuna (TPE, default settings) \\
BO trials & 100 for SNR sweep; 200 for convergence \\
\hline
\end{tabular}
\end{table}

Omnidirectional antenna elements are adopted to isolate the temporal prediction behaviour from beam-pattern and beam-alignment effects, allowing the evaluation to focus on the recurrent predictor's ability to learn mobility- and Doppler-induced temporal channel variations under noisy $2 \times 2$ MIMO conditions.

\subsection{Prediction Target and Time-Scale Interpretation}
\label{sec:prediction_target}

Although the general preprocessing pipeline supports real--imaginary frequency-domain channel features, the numerical evaluation focuses on a narrowband $2\times2$ MIMO magnitude-prediction setting. Let $\mathbf{H}(t)\in\mathbb{C}^{2\times2}$ denote the underlying MIMO channel matrix and let
\begin{equation}
\widetilde{\mathbf{H}}(t)
=
\mathbf{H}(t)+\mathbf{N}(t)
\end{equation}
denote its noisy observation at the received SNR associated with the experiment. The observed feature vector is formed by stacking the magnitudes of the four transmit--receive coefficients:
\begin{equation}
\mathbf{y}_t
=
\left[
\left|\widetilde{H}_{1,1}(t)\right|,
\left|\widetilde{H}_{1,2}(t)\right|,
\left|\widetilde{H}_{2,1}(t)\right|,
\left|\widetilde{H}_{2,2}(t)\right|
\right]^{\top}
\in\mathbb{R}^{4}.
\label{eq:observed_mimo_features}
\end{equation}
The predictor performs one-step-ahead prediction of the next noisy channel-magnitude observation:
\begin{equation}
\widehat{\mathbf{y}}_{t+1}
=
f_{\boldsymbol{\theta}}
\left(
\mathbf{y}_{t-T_{\mathrm{seq}}+1},
\ldots,
\mathbf{y}_{t}
\right).
\label{eq:one_step_prediction}
\end{equation}
For the received-SNR sweep, the observation SNR varies from $-15$ to $20$~dB, whereas $10$~dB is used for the fixed-SNR Bayesian-optimisation convergence study.

The CDL traces are sampled at $f_s = 15$\,kHz (numerology $\mu=0$ in 5G NR), corresponding to a temporal spacing of $T_s \approx 66.7\,\mu$s between consecutive channel snapshots. With carrier frequency $f_c = 3.5$\,GHz and user velocity $v = 30$\,m/s, the maximum Doppler frequency is
\begin{equation}
    f_D = \frac{v f_c}{c} \approx 350~\mathrm{Hz},
\end{equation}
where $c$ is the speed of light. The corresponding channel coherence time is approximately
\begin{equation}
    T_c \approx \frac{0.423}{f_D} \approx 1.2~\mathrm{ms},
\end{equation}
spanning approximately $T_c/T_s \approx 18$ sampled channel snapshots. This confirms that the sequence-length range explored in Bayesian optimisation is physically meaningful, covering a substantial fraction of the channel coherence interval under the considered mobility conditions. The slot duration under $\mu=0$ is 1\,ms, which is comparable to the estimated coherence time, so the channel may vary significantly between consecutive slots, making predictive channel tracking beneficial for adaptive transmission schemes.

\subsection{Prediction Accuracy versus SNR}
\label{subsec:nmse_snr_results}

We first evaluate the prediction accuracy of all architectures over a wide range of received SNR values. For each received-SNR value, $50$ independent Bayesian-optimisation runs are performed using a budget of $100$ trials per run. The best validation NMSE from each run is retained, and the reported curves show the sample mean and $95\%$ confidence interval of these $50$ best-run values.

\begin{figure}[!t]
    \centering
    \includegraphics[width=0.99\columnwidth]{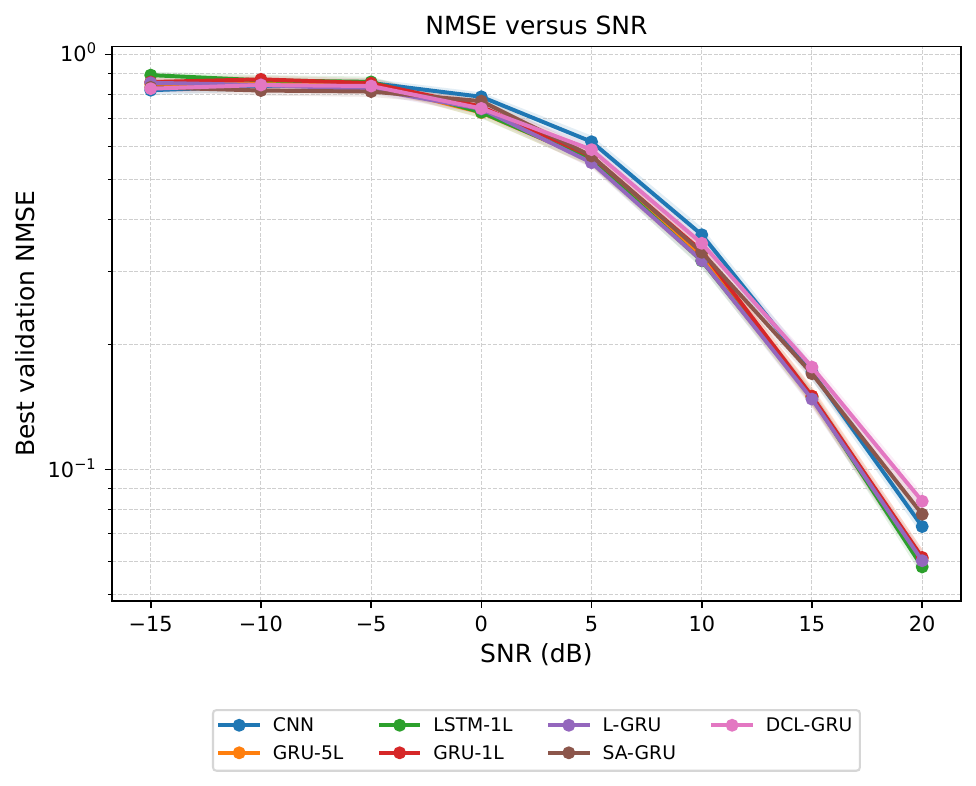}
    \caption{Best validation NMSE versus received SNR for all considered architectures. Each curve is averaged over 50 independent Bayesian-optimisation runs; shaded regions indicate 95\% confidence intervals.}
    \label{fig:nmse_vs_snr}
\end{figure}

Fig.~\ref{fig:nmse_vs_snr} shows that the NMSE decreases consistently as the received SNR increases, confirming that all models exploit the improved channel observation quality. In the low-SNR regime, the curves are close because the prediction task is noise limited. As the SNR increases, architectural differences become more visible: the unconstrained recurrent baselines---particularly LSTM, GRU-5L, and the unconstrained \acronym---achieve the lowest NMSE at high SNR. The proposed \acronymA\ and \acronymB\ remain competitive across the considered SNR range, although the imposed spectral constraints introduce a measurable accuracy penalty in the high-SNR regime. This reflects the intended accuracy--stability trade-off: the constrained variants sacrifice prediction accuracy in exchange for explicit control of the recurrent dynamics. At high received SNR, the observations contain less noise and the unconstrained recurrent models can exploit finer temporal variations through larger effective recurrent gains. The spectral bounds restrict this flexibility, so their accuracy cost becomes more visible once observation noise is no longer the dominant source of prediction error; at low SNR, this difference is largely masked by the noise-limited regime.

\begin{table}[!t]
    \centering
   \caption{Best validation NMSE at selected received-SNR values, reported as sample mean $\pm$ $95\%$ confidence-interval half-width over $50$ independent Bayesian-optimisation runs.}
    \label{tab:nmse_keypoints}
    \resizebox{\linewidth}{!}{%
        \begin{tabular}{lllll}
\toprule
Model & -15 dB & 0 dB & 10 dB & 20 dB \\
\midrule
CNN & 0.8205 ± 0.023 & 0.7912 ± 0.02 & 0.368 ± 0.015 & 0.07288 ± 0.0034 \\
GRU-5L & 0.8488 ± 0.027 & 0.7252 ± 0.021 & 0.3264 ± 0.013 & 0.06099 ± 0.0033 \\
LSTM-1L & 0.8924 ± 0.017 & 0.726 ± 0.024 & 0.3183 ± 0.013 & 0.05826 ± 0.0028 \\
GRU-1L & 0.8587 ± 0.027 & 0.7493 ± 0.024 & 0.3361 ± 0.012 & 0.06138 ± 0.003 \\
L-GRU & 0.8552 ± 0.02 & 0.7397 ± 0.023 & 0.3189 ± 0.011 & 0.06029 ± 0.003 \\
SA-GRU & 0.8339 ± 0.019 & 0.7725 ± 0.022 & 0.3338 ± 0.012 & 0.07803 ± 0.0041 \\
DCL-GRU & 0.8277 ± 0.019 & 0.7417 ± 0.022 & 0.3509 ± 0.011 & 0.08385 ± 0.0037 \\
\bottomrule
\end{tabular}

    }
\end{table}

Table~\ref{tab:nmse_keypoints} reports representative numerical values from the SNR sweep. At 10\,dB, the proposed constrained variants remain within the same NMSE range as the lightweight and recurrent baselines. At 20\,dB, the accuracy gap between the constrained and unconstrained models becomes clearer, indicating that the spectral constraints primarily affect performance in the high-SNR regime where observation noise is no longer the dominant limitation.

\subsection{Bayesian-Optimisation Convergence}
\label{subsec:bo_convergence_results}

We next evaluate how the Optuna trial budget affects the best validation NMSE. At a fixed received SNR of $10$~dB, $100$ independent Bayesian-optimisation runs are executed for $200$ trials each. For every run, the best-so-far validation NMSE is recorded at trial budgets of $50$, $100$, $150$, and $200$.

\begin{figure}[!t]
    \centering
    \includegraphics[width=0.99\columnwidth]{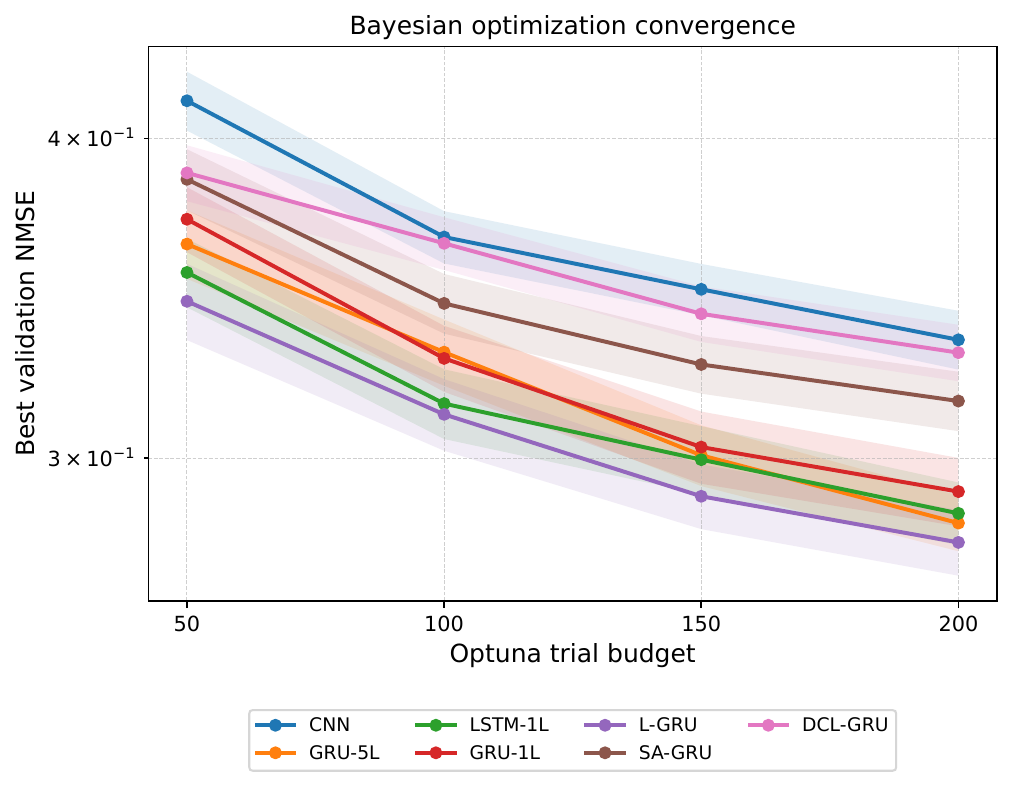}
    \caption{Bayesian-optimisation convergence at $10$~dB. Each point shows the mean best-so-far validation NMSE over $100$ independent runs, and the shaded region denotes the corresponding $95\%$ confidence interval.}
    \label{fig:bo_convergence}
\end{figure}

Fig.~\ref{fig:bo_convergence} shows that increasing the trial budget improves the best validation NMSE for all architectures. Most of the improvement occurs between 50 and 150 trials, after which gains become more gradual. The constrained SA-GRU and DCL-GRU exhibit smooth best-so-far convergence trends comparable to those of the unconstrained lightweight recurrent models. L-GRU obtains the lowest final NMSE in this fixed-SNR experiment, while \acronymB\ provides the strongest structural guarantee among the lightweight variants.

\subsection{Computational Runtime}
\label{subsec:runtime_results}

In addition to prediction accuracy, computational cost is a key metric for real-time PHY-layer deployment. We compare the cumulative optimisation time as a function of the BO trial budget and the average wall-clock time per SNR-sweep run.

\begin{figure}[!t]
    \centering
    \includegraphics[width=0.99\columnwidth]{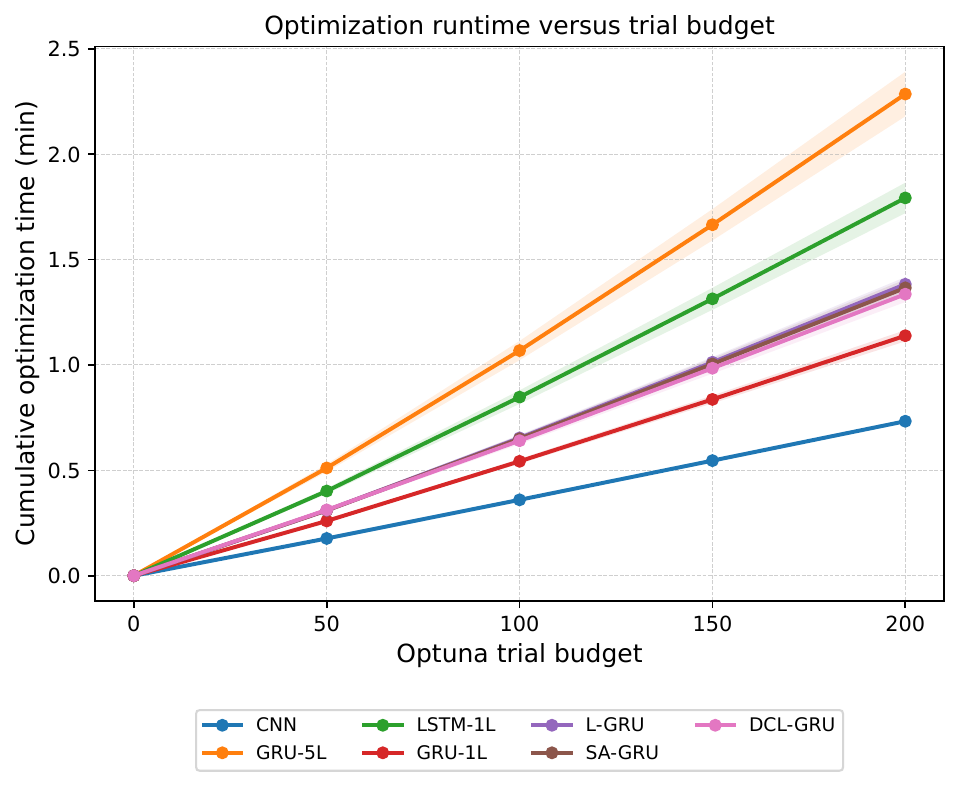}
    \caption{Cumulative optimisation time versus Optuna trial budget. Lightweight recurrent variants require substantially less optimisation time than the deeper GRU baseline.}
    \label{fig:runtime_vs_trials}
\end{figure}

Fig.~\ref{fig:runtime_vs_trials} shows that optimisation time grows approximately linearly with the number of Optuna trials, as expected since each trial requires a complete training and validation cycle. The GRU-5L baseline is the most computationally expensive architecture, whereas the lightweight recurrent variants exhibit substantially lower cumulative optimisation time.

\begin{figure}[!t]
    \centering
    \includegraphics[width=\linewidth]{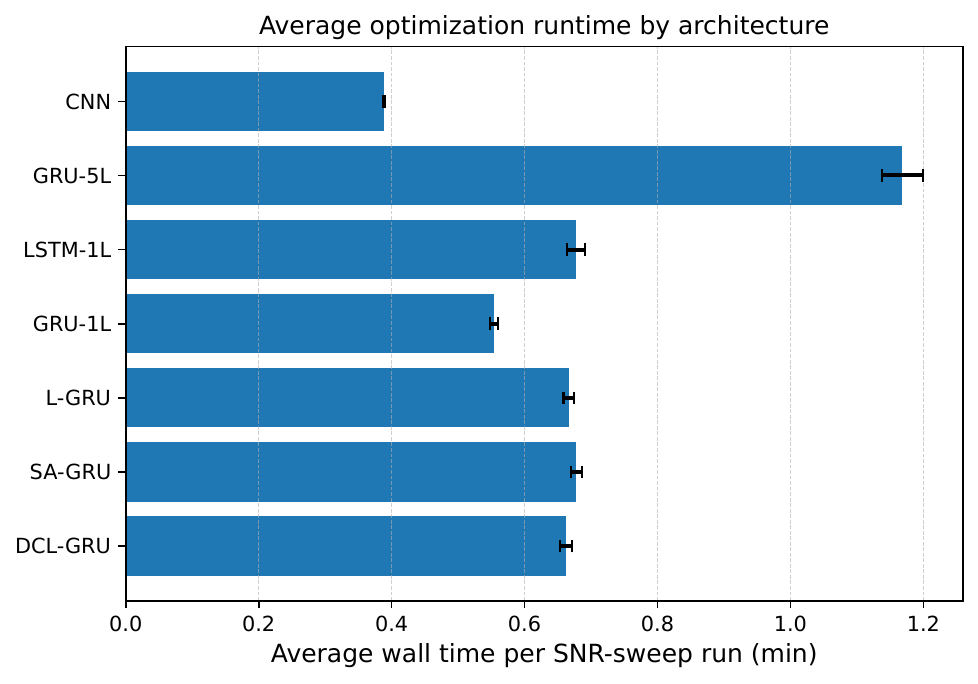}
    \caption{Average wall-clock optimisation time per SNR-sweep run for each architecture. The proposed constrained variants retain a runtime close to the unconstrained lightweight GRU, confirming that the post-training spectral projections introduce negligible overhead.}
    \label{fig:runtime_by_model}
\end{figure}

\begin{table}[!t]
\centering
\caption{Average optimisation runtime by architecture, reported as sample mean $\pm$ $95\%$ confidence-interval half-width over $400$ model--SNR optimisation runs. Speedup is measured relative to GRU-5L.}
\label{tab:runtime_by_model}
\scriptsize
\begin{tabular}{lccc}
\toprule
Model & Runs & Mean time (min) & Speedup \\
\midrule
CNN      & 400 & $0.389 \pm 0.0015$ & $3.00\times$ \\
GRU-5L   & 400 & $1.169 \pm 0.0307$ & $1.00\times$ \\
LSTM-1L  & 400 & $0.677 \pm 0.0132$ & $1.73\times$ \\
GRU-1L   & 400 & $0.554 \pm 0.0055$ & $2.11\times$ \\
L-GRU    & 400 & $0.667 \pm 0.0082$ & $1.75\times$ \\
SA-GRU   & 400 & $0.678 \pm 0.0082$ & $1.72\times$ \\
DCL-GRU  & 400 & $0.662 \pm 0.0088$ & $1.76\times$ \\
\bottomrule
\end{tabular}
\end{table}

Fig.~\ref{fig:runtime_by_model} and Table~\ref{tab:runtime_by_model} confirm that the proposed \acronymA\ and \acronymB\ have runtimes close to the unconstrained \acronym. The stability-aware and candidate-state contractive projections therefore introduce no meaningful optimisation-time penalty; the runtime of the constrained variants is dominated by the underlying recurrent training loop rather than by the one-time spectral projection.

\subsection{Projection and Contraction Audit}
\label{subsec:contractivity_audit_results}

We verify that the empirical results correspond to the deployed constrained models. For \acronymA, the audit checks whether the projected candidate recurrent matrix satisfies $\|\mathbf{U}_h\|_2\le\rho_h$. For \acronymB, the audit additionally checks $\|\mathbf{U}_r\|_2\le\rho_r$ and the sufficient candidate-state contraction condition
\begin{equation}
\rho_h\!\left(1+\frac{\rho_r}{4}\right)\le 1-\delta.
\end{equation}

\begin{table*}[!t]
\centering
\caption{Projection and contraction audit for the proposed constrained variants. The table verifies that all deployed constrained models satisfy the required spectral-norm bounds and, for DCL-GRU, the sufficient candidate-state contraction condition.}
\label{tab:contractivity_audit}
\scriptsize
\begin{tabular}{llcccccc}
\toprule
Experiment & Model 
& $U_h$ checks 
& max norm / bound for $\mathbf{U}_h$
& $U_r$ checks 
& max norm / bound for $\mathbf{U}_r$
& Condition checks 
& Violations \\
\midrule
Convergence & DCL-GRU & 20000 & $0.84/0.84$ & 20000 & $0.50/0.50$ & 20000 & 0 \\
Convergence & SA-GRU  & 20000 & $0.90/0.90$ & --    & --          & --    & 0 \\
SNR sweep   & DCL-GRU & 40000 & $0.84/0.84$ & 40000 & $0.50/0.50$ & 40000 & 0 \\
SNR sweep   & SA-GRU  & 40000 & $0.90/0.90$ & --    & --          & --    & 0 \\
\bottomrule
\end{tabular}
\end{table*}

Table~\ref{tab:contractivity_audit} shows that all audited constrained runs satisfy the required spectral-norm bounds, with zero violations. In particular, \acronymB\ satisfies the sufficient candidate-state contraction condition in every audited run. This confirms that the reported accuracy and runtime results correspond to models whose deployed recurrent matrices obey the intended structural constraints.

\subsection{Empirical Stability Under Corrupted Recursive Rollouts}
\label{sec:corrupted_rollout}

The one-step prediction results reported previously quantify model accuracy when predictions are formed from observed channel samples. However, one-step metrics do not fully reveal how recurrent predictors respond to temporarily corrupted observations or how the resulting perturbations propagate when predictions are recursively fed back over multiple time steps. To investigate this behaviour, we compare the unconstrained L-GRU with the SA-GRU and DCL-GRU variants under a temporary observation corruption followed by open-loop recursive prediction.

The purpose of this experiment is not to claim contraction of the complete GRU hidden-state transition. Instead, it evaluates whether the spectral constraints imposed on the candidate-state recurrent pathways reduce the empirical sensitivity of the recurrent state to corrupted inputs and quantifies the corresponding effect on recursive prediction accuracy.

\subsubsection{Evaluation Protocol}

All three variants are initialised from the same pretrained L-GRU parameters. A common base model is first trained for 15 epochs using 50 independent trajectories of 250 samples each. The trained model is then copied to form the L-GRU, SA-GRU, and DCL-GRU variants, and each copy is fine-tuned for 10 additional epochs using one tenth of the original learning rate. For SA-GRU and DCL-GRU, the corresponding spectral projection is applied before fine-tuning and reapplied after every optimiser update. This procedure allows the constrained models to adapt to the projected recurrent dynamics while continuously satisfying the prescribed bounds.

The shared architecture uses hidden-state dimension $H=240$, sequence length $T_{\mathrm{seq}}=13$, batch size 16, dropout probability $0.4787$, and initial learning rate $4.456\times10^{-3}$. The SA-GRU constraint is $\rho_h=0.90$. For DCL-GRU, the constraint parameters are $\rho_h=0.84$, $\rho_r=0.50$, and $\delta=0.05$, which satisfy
\begin{equation}
\rho_h\left(1+\frac{\rho_r}{4}\right)=0.945\leq 1-\delta=0.95.
\label{eq:rollout_dcl_condition}
\end{equation}

For each test run, a clean trajectory of $T_{\mathrm{test}}=200$ channel snapshots is generated using the same $2\times2$ MIMO CDL-A configuration described in Section~\ref{sec:numerical_evaluations}. Let $\mathbf{x}_t\in\mathbb{R}^{F}$, with $F=4$, denote the underlying clean channel-magnitude feature vector at time index $t$.

A baseline noisy observation sequence is generated according to
\begin{equation}
\mathbf{y}_t=\mathbf{x}_t+\mathbf{w}_t,
\label{eq:baseline_noisy_observation}
\end{equation}
where $\mathbf{w}_t$ denotes additive observation noise corresponding to an SNR of $\gamma_{\mathrm{o}}=10$~dB, matching the observation condition used during training.

An additional temporary feature-domain corruption is applied to the baseline noisy observations according to
\begin{equation}
\widetilde{\mathbf{y}}_t=
\begin{cases}
\mathbf{y}_t+\mathbf{n}_t, & t_{\mathrm{c}}\leq t<t_0,\\[1mm]
\mathbf{y}_t, & \text{otherwise},
\end{cases}
\label{eq:corrupted_input}
\end{equation}
where $t_0=t_{\mathrm{c}}+L_{\mathrm{c}}$, with $t_{\mathrm{c}}=40$ and $L_{\mathrm{c}}=10$. The additional corruption vector is generated as
\begin{equation}
\mathbf{n}_t\sim\mathcal{N}\left(\mathbf{0},\sigma_{\mathrm{c}}^2\mathbf{I}_{F}\right),
\qquad
\sigma_{\mathrm{c}}^2=P_{\mathrm{c}}10^{-\gamma_{\mathrm{c}}/10},
\label{eq:corruption_noise}
\end{equation}
where
\begin{equation}
P_{\mathrm{c}}=\frac{1}{F L_{\mathrm{c}}}\sum_{t=t_{\mathrm{c}}}^{t_0-1}\left\|\mathbf{x}_t\right\|_2^2
\label{eq:corruption_power}
\end{equation}
is the average clean feature power over the corrupted interval and $\gamma_{\mathrm{c}}=0$~dB is the additional corruption SNR.

The clean-control and corrupted observation sequences are defined as $\mathbf{u}^{(\mathrm{clean})}_t=\mathbf{y}_t$ and
$\mathbf{u}^{(\mathrm{corr})}_t=\widetilde{\mathbf{y}}_t$, respectively.

\begin{figure*}[t]
\centering
\includegraphics[width=\textwidth]{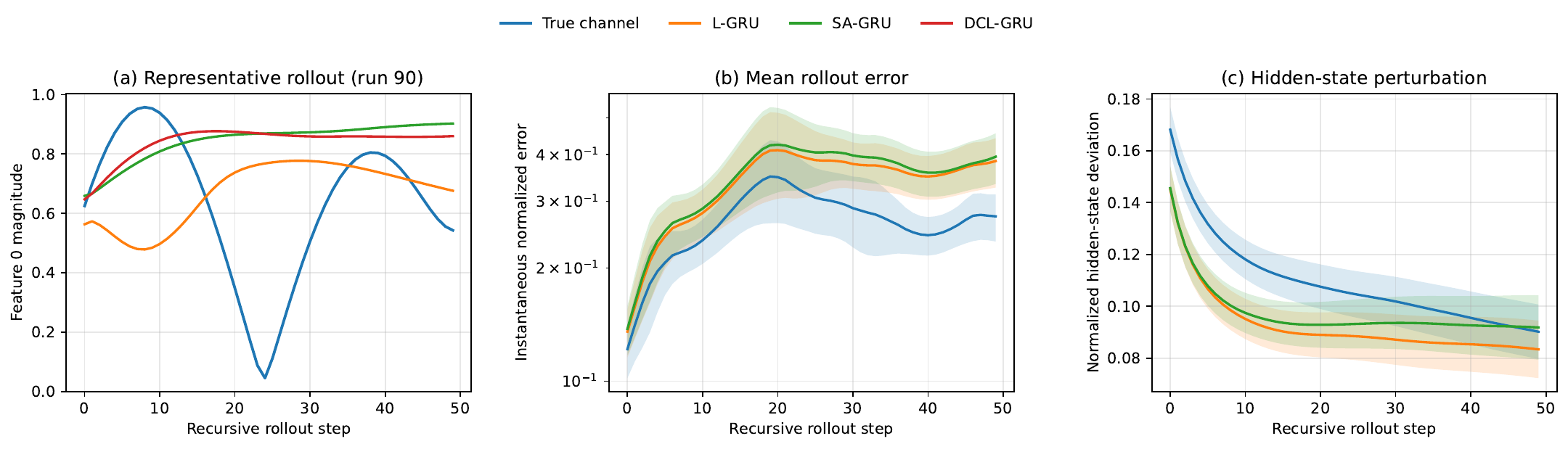}%
\caption{Empirical stability under an additional temporary observation corruption. Left: representative $50$-step corrupted recursive rollout. Centre: mean instantaneous normalised prediction error with $95\%$ confidence intervals over $100$ paired test trajectories. Right: normalised deviation between the hidden states generated by the corrupted and clean-control observation sequences. All models are evaluated using the same underlying channel trajectories, baseline-noise realisations, and corruption realisations.}
\label{fig:corrupted_rollout}
\end{figure*}

Both evaluations therefore contain the same baseline observation noise, while only the corrupted evaluation contains the additional temporary burst. The same underlying channel trajectory, baseline-noise realisation, and corruption realisation are used for all three models.

Both observation sequences are transformed using the same normalisation parameters computed from the training set. Recursive feedback is performed in the normalised feature domain used during training, while predictions are transformed back to the original feature domain before the reported metrics are calculated.

Let $q\in\{\mathrm{clean},\mathrm{corr}\}$ denote the clean-control and corrupted evaluations, respectively. For model $m$, the observation prefix is processed sequentially according to
\begin{equation}
\mathbf{s}^{(m,q)}_t=\mathcal{G}_{\theta_m}\left(\mathbf{u}^{(q)}_t,\mathbf{s}^{(m,q)}_{t-1}\right),
\qquad
t<t_0,
\label{eq:rollout_warmup}
\end{equation}
where $\mathcal{G}_{\theta_m}$ denotes the recurrent state-update map. The two evaluations use the same zero-state initialisation, but their hidden states can differ after the corrupted interval has been processed.

At $t_0=50$, external observations are removed and each model is operated recursively for $H_{\mathrm{roll}}=50$ steps. For $k=0,\ldots,H_{\mathrm{roll}}-1$, the model output is generated as
\begin{equation}
\widehat{\mathbf{x}}^{(m,q)}_{t_0+k}=\mathcal{R}_{\theta_m}\left(\mathbf{s}^{(m,q)}_{t_0+k-1}\right),
\label{eq:recursive_readout}
\end{equation}
where $\mathcal{R}_{\theta_m}$ denotes the output readout.

The recurrent state is subsequently updated according to
\begin{equation}
\mathbf{s}^{(m,q)}_{t_0+k}=\mathcal{G}_{\theta_m}\left(\widehat{\mathbf{x}}^{(m,q)}_{t_0+k},\mathbf{s}^{(m,q)}_{t_0+k-1}\right).
\label{eq:recursive_state}
\end{equation}

Hence, after $t_0$, each prediction is used as the next model input while the recurrent hidden state is propagated continuously.

The instantaneous normalised squared prediction error of the corrupted rollout is defined as
\begin{equation}
E_m(k)=
\frac{
\left\|
\widehat{\mathbf{x}}^{(m,\mathrm{corr})}_{t_0+k}
-
\mathbf{x}_{t_0+k}
\right\|_2^2
}{
\left\|\mathbf{x}_{t_0+k}\right\|_2^2+\varepsilon_E
}.
\label{eq:instantaneous_rollout_error}
\end{equation}

The aggregate rollout NMSE is defined as
\begin{equation}
\mathrm{NMSE}^{\mathrm{roll}}_m=
\frac{
\displaystyle\sum_{k=0}^{H_{\mathrm{roll}}-1}
\left\|
\widehat{\mathbf{x}}^{(m,\mathrm{corr})}_{t_0+k}
-
\mathbf{x}_{t_0+k}
\right\|_2^2
}{
\displaystyle\sum_{k=0}^{H_{\mathrm{roll}}-1}
\left\|\mathbf{x}_{t_0+k}\right\|_2^2+\varepsilon_E
}.
\label{eq:aggregate_rollout_nmse}
\end{equation}

To isolate the effect of the additional corruption from ordinary multi-step prediction error, the output deviation between the corrupted and clean-control rollouts is defined as
\begin{equation}
D_m(k)=
\frac{
\left\|
\widehat{\mathbf{x}}^{(m,\mathrm{corr})}_{t_0+k}
-
\widehat{\mathbf{x}}^{(m,\mathrm{clean})}_{t_0+k}
\right\|_2
}{
\left\|\mathbf{x}_{t_0+k}\right\|_2+\varepsilon_D
}.
\label{eq:rollout_output_deviation}
\end{equation}

The common ground-truth denominator permits direct comparison across models and avoids artificial amplification when the magnitude of a clean-control prediction approaches zero.

The normalised hidden-state deviation is defined as
\begin{equation}
\Delta_m(k)=
\frac{
\left\|
\mathbf{s}^{(m,\mathrm{corr})}_{t_0+k}
-
\mathbf{s}^{(m,\mathrm{clean})}_{t_0+k}
\right\|_2
}{
\sqrt{H_m}
},
\label{eq:hidden_state_deviation}
\end{equation}
where $H_m$ is the hidden-state dimension of model $m$.

In addition to the peak and terminal hidden-state deviations, the average hidden-state deviation is reported as
\begin{equation}
\overline{\Delta}_m
=
\frac{1}{H_{\mathrm{roll}}}
\sum_{k=0}^{H_{\mathrm{roll}}-1}
\Delta_m(k).
\label{eq:mean_hidden_state_deviation}
\end{equation}

For a quantity $z$ evaluated over $N_{\mathrm{run}}=100$ paired test trajectories, the reported $95\%$ confidence interval is
\begin{equation}
\operatorname{CI}_{95\%}(z)=
\overline{z}
\pm
t_{0.975,N_{\mathrm{run}}-1}
\frac{s_z}{\sqrt{N_{\mathrm{run}}}},
\label{eq:rollout_confidence_interval}
\end{equation}
where $\overline{z}$ and $s_z$ are the sample mean and sample standard deviation, respectively.

These confidence intervals quantify variability across channel trajectories, baseline-noise realisations, and corruption realisations conditional on the fixed trained models. They do not quantify variability across independent model-training repetitions.

\subsubsection{Results}

\begin{table*}[t]
\centering
\caption{Summary of the corrupted recursive-rollout experiment. Values are reported as the sample mean and $95\%$ confidence-interval half-width over $100$ paired test trajectories. Bold values indicate the smallest numerical mean for each metric.}
\label{tab:corrupted_rollout_summary}
\scriptsize
\setlength{\tabcolsep}{3.2pt}
\begin{tabular}{lcccccc}
\toprule
Model & Rollout NMSE & Peak output dev. & Terminal output dev. & Mean hidden dev. & Peak hidden dev. & Terminal hidden dev. \\
\midrule
L-GRU & $\mathbf{0.2415\pm0.0180}$ & $0.3711\pm0.0353$ & $\mathbf{0.2172\pm0.0364}$ & $0.1094\pm0.0083$ & $0.1706\pm0.0084$ & $0.0902\pm0.0105$ \\
SA-GRU & $0.3141\pm0.0310$ & $\mathbf{0.3535\pm0.0294}$ & $0.2310\pm0.0313$ & $\mathbf{0.0927\pm0.0086}$ & $\mathbf{0.1484\pm0.0088}$ & $\mathbf{0.0834\pm0.0111}$ \\
DCL-GRU & $0.3225\pm0.0331$ & $0.3813\pm0.0335$ & $0.2567\pm0.0367$ & $0.0975\pm0.0089$ & $0.1518\pm0.0092$ & $0.0918\pm0.0124$ \\
\bottomrule
\end{tabular}
\end{table*}

Figure~\ref{fig:corrupted_rollout} shows that the unconstrained L-GRU maintains the lowest instantaneous normalised prediction error over most of the recursive horizon. Consistently, Table~\ref{tab:corrupted_rollout_summary} reports an aggregate rollout NMSE of $0.2415\pm0.0180$ for L-GRU, compared with $0.3141\pm0.0310$ for SA-GRU and $0.3225\pm0.0331$ for DCL-GRU. Thus, the spectral constraints do not improve open-loop recursive prediction accuracy in this setting. Relative to L-GRU, the rollout NMSE increases by approximately $30.0\%$ for SA-GRU and $33.5\%$ for DCL-GRU.

In contrast, the constrained variants exhibit reduced propagation of the observation perturbation through their recurrent states. SA-GRU exhibits a numerically lower mean hidden-state deviation, decreasing from $0.1094\pm0.0083$ to $0.0927\pm0.0086$, corresponding to a reduction of approximately $15.3\%$. Its peak hidden-state deviation decreases from $0.1706\pm0.0084$ to $0.1484\pm0.0088$, corresponding to a reduction of approximately $13.0\%$. DCL-GRU also exhibits numerically lower mean and peak hidden-state deviations by approximately $10.9\%$ and $11.0\%$, respectively, relative to L-GRU.

SA-GRU also obtains the smallest peak output deviation, namely $0.3535\pm0.0294$, which represents a numerical reduction of approximately $4.7\%$ relative to L-GRU. Nevertheless, L-GRU achieves the smallest terminal output deviation, while the confidence intervals of the terminal hidden-state deviations overlap substantially across the three models. The terminal metrics should therefore not be interpreted as evidence of a clear ordering among the models without an additional paired statistical test.

Overall, the experiment reveals a stability--accuracy trade-off. Constraining the candidate-state recurrent pathways reduces the sensitivity of the internal recurrent state to a temporary observation corruption, but this improved state robustness is accompanied by higher recursive prediction NMSE. Under the considered configuration, SA-GRU provides the strongest empirical suppression of the hidden-state perturbation, whereas L-GRU preserves the highest open-loop prediction accuracy. DCL-GRU satisfies the sufficient candidate-state contraction condition, but its stronger constraint does not provide an additional empirical advantage over SA-GRU in this rollout experiment.

\subsection{Discussion}
\label{subsec:results_discussion}

The results show that the proposed \acronymA\ and \acronymB\ preserve competitive channel-prediction accuracy while introducing explicit spectral constraints on the recurrent dynamics. The unconstrained baselines achieve slightly lower NMSE in some high-SNR regimes, but they do not provide stability-aware or candidate-state contraction guarantees. In contrast, SA-GRU and DCL-GRU provide explicit structural guarantees on selected recurrent pathways, with DCL-GRU additionally satisfying the sufficient candidate-state contraction condition in all audited runs while maintaining a runtime close to the unconstrained lightweight GRU. The proposed constrained variants therefore offer a practical trade-off between prediction accuracy, computational efficiency, and theoretically controlled recurrent dynamics.

Interestingly, the stronger theoretical guarantee provided by \acronymB\ does not necessarily translate into superior empirical performance under the considered rollout configuration, illustrating that sufficient contraction guarantees and empirical rollout robustness need not coincide. This observation suggests that stronger structural constraints may introduce additional conservatism, highlighting a trade-off between theoretical guarantees and practical performance. A possible mechanism is that the additional bound on the reset-gate recurrent matrix in DCL-GRU limits how strongly the reset gate can adapt the candidate-state update to the recent hidden-state trajectory. Although this restriction is required by the sufficient contraction condition, it may also reduce the model's ability to apply trajectory-specific corrections after a corruption burst, explaining why the stronger DCL-GRU guarantee does not translate into lower empirical rollout error than SA-GRU.

The corrupted-rollout analysis was conducted under a single perturbation configuration, and a broader sensitivity analysis across multiple corruption scenarios, observation conditions, and contraction margins is left for future work.
% =================================
% Conclusions and Future Directions
% =================================
%\FloatBarrier
\section{Conclusions and Future Directions}
\label{sec:conclusions}

\subsection{Conclusions}

This paper presented a framework for causal one-step-ahead channel prediction using lightweight GRU-based recurrent architectures under noisy $2\times2$ MIMO CDL channels. In addition to the unconstrained lightweight GRU baseline, two constrained variants were investigated. The first, \acronymA, imposes a post-training spectral constraint on the candidate-state recurrent matrix, bounding the direct candidate recurrent gain. The second, \acronymB, additionally constrains the reset-gate recurrent matrix and satisfies a sufficient condition for contraction of the complete candidate-state mapping.

The numerical results show that the proposed constrained variants remain competitive with conventional CNN, GRU, LSTM, and lightweight recurrent baselines across a broad SNR range. The unconstrained models can achieve lower NMSE in some high-SNR settings, but the constrained variants provide explicit structural guarantees that the baselines lack. Runtime measurements further confirm that the post-training projections introduce negligible overhead relative to the unconstrained lightweight GRU, and that the lightweight recurrent variants are substantially faster to optimise than the deeper GRU baseline. The projection audit confirms that the required spectral constraints and candidate-state contraction condition are satisfied in all completed constrained runs.

The corrupted recursive-rollout experiment further demonstrated that the spectral constraints reduce the empirical propagation of observation perturbations through the recurrent state. SA-GRU exhibited mean and peak hidden-state deviations approximately $15.3\%$ and $13.0\%$ below those of L-GRU, respectively. This improved state robustness was accompanied by higher recursive rollout NMSE, demonstrating that the proposed spectral constraints introduce a measurable stability--accuracy trade-off rather than an across-the-board improvement in prediction performance.

\subsection{Future Directions}

Several natural extensions are identified. First, future work will consider wideband OFDM channels with multiple subcarriers, realistic pilot placement, and DMRS-based channel observations. Second, online and continual learning mechanisms may be combined with the proposed spectrally constrained recurrent structure to improve adaptation under long-term distribution shift while preserving the projected recurrent constraints. Third, model compression techniques such as pruning, low-precision quantisation, and FPGA-oriented implementation can be explored to further reduce memory footprint and inference latency. Finally, integrating the predictor into link adaptation, scheduling, or goal-oriented communication loops would enable evaluation of the system-level impact of stable channel prediction.

\appendices

% =================================
% Proof of Proposition 1
% =================================
\section{Proof of Proposition~\ref{prop:acronym_spectral_bound}}
\label{appendix:A}

The hyperbolic tangent is globally $1$-Lipschitz. For any scalar $a\in\mathbb{R}$,
\begin{equation}
\frac{\mathrm{d}}{\mathrm{d}a}\tanh(a) = 1-\tanh^2(a).
\end{equation}
Since $|\tanh(a)|<1$ for all finite $a$, we have $0 < 1-\tanh^2(a) \le 1$. By the mean-value theorem, for any $a,b\in\mathbb{R}$,
\begin{equation}
|\tanh(a)-\tanh(b)| \le |a-b|.
\end{equation}
When $\tanh(\cdot)$ is applied element-wise to vectors $\mathbf{a},\mathbf{b}\in\mathbb{R}^{H}$,
\begin{align}
\left\|\tanh(\mathbf{a}) - \tanh(\mathbf{b})\right\|_2^2
&= \sum_{i=1}^{H}|\tanh(a_i)-\tanh(b_i)|^2 \nonumber \\
&\le \|\mathbf{a}-\mathbf{b}\|_2^2,
\end{align}
so the element-wise hyperbolic tangent is globally $1$-Lipschitz in the Euclidean norm.

For a fixed input $\mathbf{x}_t$ and a fixed reset-gate realisation $\overline{\mathbf{r}}_t\in(0,1)^H$, the conditional candidate-state mapping is
\begin{equation}
\mathcal{F}_t(\mathbf{h})
= \tanh\!\left(\mathbf{W}_h\mathbf{x}_t
+ \mathbf{U}_h\operatorname{diag}(\overline{\mathbf{r}}_t)\mathbf{h}
+ \mathbf{b}_h\right).
\end{equation}
By the chain rule and the submultiplicative property of the spectral norm,
\begin{equation}
\left\|\frac{\partial\widetilde{\mathbf{h}}_t}{\partial\mathbf{h}_{t-1}}\right\|_2
\le
\underbrace{\|\mathbf{D}_{h,t}\|_2}_{\le\,1}
\cdot
\|\mathbf{U}_h\|_2
\cdot
\underbrace{\|\operatorname{diag}(\overline{\mathbf{r}}_t)\|_2}_{\le\,1}
\le
\|\mathbf{U}_h\|_2,
\label{eq:conditional_candidate_jacobian_bound}
\end{equation}
where $\mathbf{D}_{h,t} = \operatorname{diag}(\mathbf{1} - \widetilde{\mathbf{h}}_t\odot\widetilde{\mathbf{h}}_t)$. Both bounding factors are at most 1 because $|\widetilde{h}_{t,i}|<1$ and $|\overline{r}_{t,i}|<1$ for all $i$. Therefore, enforcing $\|\mathbf{U}_h\|_2\le\rho_h$ via the projection in~\eqref{eq:uh_rescaling} guarantees that the Lipschitz constant of the conditional candidate-state mapping with respect to $\mathbf{h}_{t-1}$ is at most $\rho_h$.\qed

% =================================
% Proof of Proposition 2
% =================================
\section{Proof of Proposition~\ref{prop:complete_candidate_contraction}}
\label{appendix:B}

\emph{Boundedness of the hidden state.} Assume that the recurrent state is initialised such that $\|\mathbf{h}_0\|_{\infty}\leq1$, as satisfied by the zero initialisation used in the experiments. Since the hyperbolic tangent is applied element-wise, every component of the candidate state satisfies $|\widetilde{h}_{t,i}|<1$, so $\|\widetilde{\mathbf{h}}_t\|_\infty\le 1$. The complete GRU hidden-state update is
\begin{equation}
\mathbf{h}_t = (\mathbf{1}-\mathbf{z}_t)\odot\mathbf{h}_{t-1} + \mathbf{z}_t\odot\widetilde{\mathbf{h}}_t,
\end{equation}
where $0<z_{t,i}<1$ for every $i$, so each component $h_{t,i}$ is a convex combination of $h_{t-1,i}$ and $\widetilde{h}_{t,i}$. By induction, $\|\mathbf{h}_t\|_\infty\le 1$ for all $t\ge 0$.

\emph{Jacobian of the complete candidate-state mapping.} Defining the reset-gate pre-activation $\mathbf{a}_{r,t} = \mathbf{W}_r\mathbf{x}_t + \mathbf{U}_r\mathbf{h}_{t-1} + \mathbf{b}_r$ and $\mathbf{r}_t=\sigma(\mathbf{a}_{r,t})$, the chain rule gives
\begin{equation}
\frac{\partial\mathbf{r}_t}{\partial\mathbf{h}_{t-1}} = \mathbf{D}_{r,t}\mathbf{U}_r,
\end{equation}
where $\mathbf{D}_{r,t}=\operatorname{diag}(\mathbf{r}_t\odot(\mathbf{1}-\mathbf{r}_t))$. Since the sigmoid derivative satisfies $0<\sigma'(a)\le\tfrac{1}{4}$, we have $\|\mathbf{D}_{r,t}\|_2\le\tfrac{1}{4}$. The partial derivative of the gated term $\mathbf{r}_t\odot\mathbf{h}_{t-1}$ with respect to $\mathbf{h}_{t-1}$ is
\begin{equation}
\frac{\partial(\mathbf{r}_t\odot\mathbf{h}_{t-1})}{\partial\mathbf{h}_{t-1}}
= \operatorname{diag}(\mathbf{r}_t) + \operatorname{diag}(\mathbf{h}_{t-1})\mathbf{D}_{r,t}\mathbf{U}_r.
\end{equation}
Applying the chain rule to the candidate-state mapping and taking spectral norms yields
\begin{align}
\left\|\mathbf{J}_{\widetilde{h},t}\right\|_2 \le & \|\mathbf{D}_{h,t}\|_2\,\|\mathbf{U}_h\|_2 \!\big[\|\operatorname{diag}(\mathbf{r}_t)\|_2 \nonumber \\
&+ \|\operatorname{diag}(\mathbf{h}_{t-1})\|_2\,\|\mathbf{D}_{r,t}\|_2\,\|\mathbf{U}_r\|_2\big]\\
\le & \|\mathbf{U}_h\|_2\!\left(1+\tfrac{1}{4}\|\mathbf{U}_r\|_2\right)
\le \rho_h\!\left(1+\tfrac{\rho_r}{4}\right)
\le 1-\delta < 1, \nonumber
\end{align}
where we used $\|\mathbf{D}_{h,t}\|_2\le 1$, $\|\operatorname{diag}(\mathbf{r}_t)\|_2\le 1$, $\|\operatorname{diag}(\mathbf{h}_{t-1})\|_2\le 1$, and the assumed spectral-norm bounds and contraction condition~\eqref{eq:dcl_contraction_condition}. Hence, the candidate-state mapping is globally contractive on $\{\mathbf{h}:\|\mathbf{h}\|_\infty\le 1\}$ with contraction factor at most $1-\delta$.\qed

% =================================
% Bibliography
% =================================
\bibliographystyle{IEEEtran}
\bibliography{references}

% ========================================
%
% IEEEbiography
%
% ========================================
\begin{IEEEbiography}
[{\includegraphics[width=1in,height=1.25in,clip,keepaspectratio]{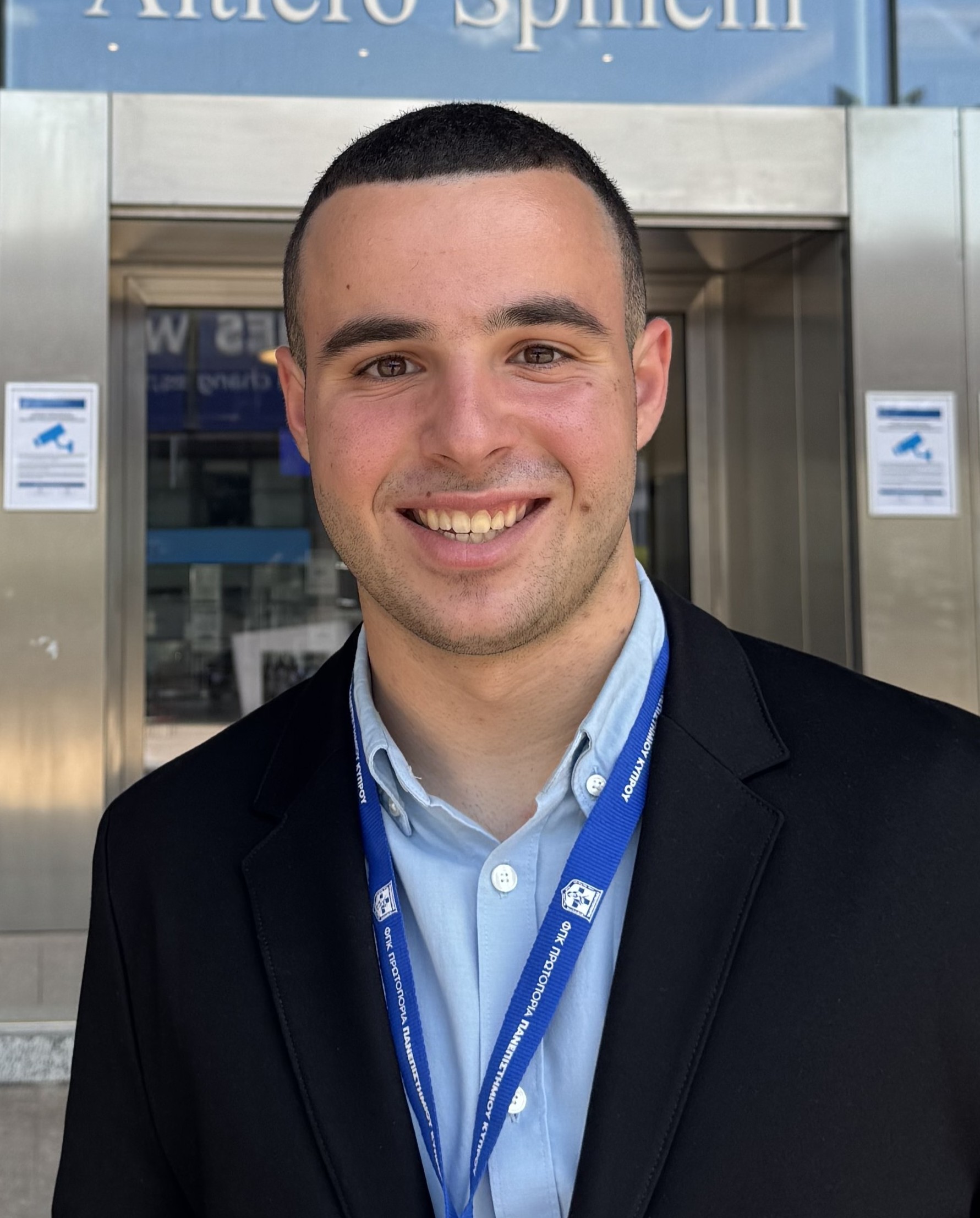}}]{Kyriakos Christodoulides}
is currently pursuing the B.Sc. degree in Electrical and Computer Engineering with the University of Cyprus, Nicosia, Cyprus. He is a Research Assistant with the Distributed and Networked Control Systems (DNCS) Group, Department of Electrical and Computer Engineering, University of Cyprus. He previously worked as a Research Assistant with the IRIDA Research Centre, University of Cyprus, during 2025--2026.       
His research interests include wireless communications, channel prediction, machine learning for communication systems, recurrent neural networks, reconfigurable intelligent surfaces, and physical-layer signal processing. His current research focuses on lightweight and stability-aware recurrent neural networks for real-time channel prediction in time-varying MIMO wireless channels.
\end{IEEEbiography}

\begin{IEEEbiography}[{\includegraphics[width=1in,height=1.25in,clip,keepaspectratio]{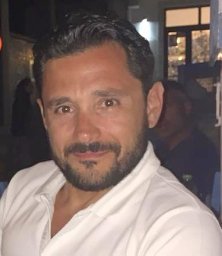}}]{Kyriakos M. Deliparaschos} (Senior Member, IEEE) received his B.Eng. (Hons) in Electronics Engineering and M.Sc. in Mechatronics from De Montfort University, UK, and the National Technical University of Athens (NTUA), Greece. He earned his Ph.D. from the Department of Signals, Control and Robotics, School of Electrical and Computer Engineering, NTUA. He is currently a Special Teaching Staff member in the Department of Electrical and Computer Engineering and Informatics at the Cyprus University of Technology (CUT). He has held academic and research positions at New York College and IST College (Athens), NTUA, CUT, Trinity College Dublin, and Cranfield University. His research interests include embedded systems, reconfigurable hardware acceleration, sensor fusion, resilient cyber-physical systems, intelligent control, real-time robotic navigation, and medical robotics for minimally invasive surgery.
\end{IEEEbiography}

\begin{IEEEbiography}
[{\includegraphics[width=1in,height=1.25in,clip,keepaspectratio]{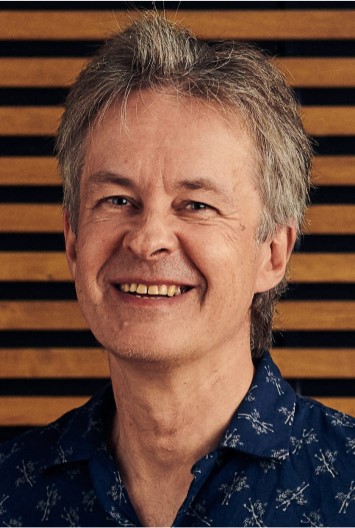}}] {Risto Wichman}~(Fellow, IEEE) received his M.Sc.  and D.Sc. (Tech) degrees in digital signal processing from Tampere University of Technology, Finland, in 1990 and 1995, respectively.  Following his doctoral studies, he joined Nokia Research Center, where he worked from 1995 to 2001 as a Senior Research Engineer. During this period, he gained substantial industrial research experience in the development of advanced signal processing techniques for wireless communication systems  In 2002, he joined the Department of Information and Communications Engineering at the Aalto University School of Electrical Engineering, Finland. Since 2008, he was appointed as a Full Professor.  His research
interests are in signal processing for wireless communication systems. His work builds on tools from communication theory, detection and estimation theory, stochastic geometry, random matrix theory, linear algebra,
and large-system analysis.
\end{IEEEbiography}

\begin{IEEEbiography}[{\includegraphics[width=1in,height=1.25in,clip,keepaspectratio]{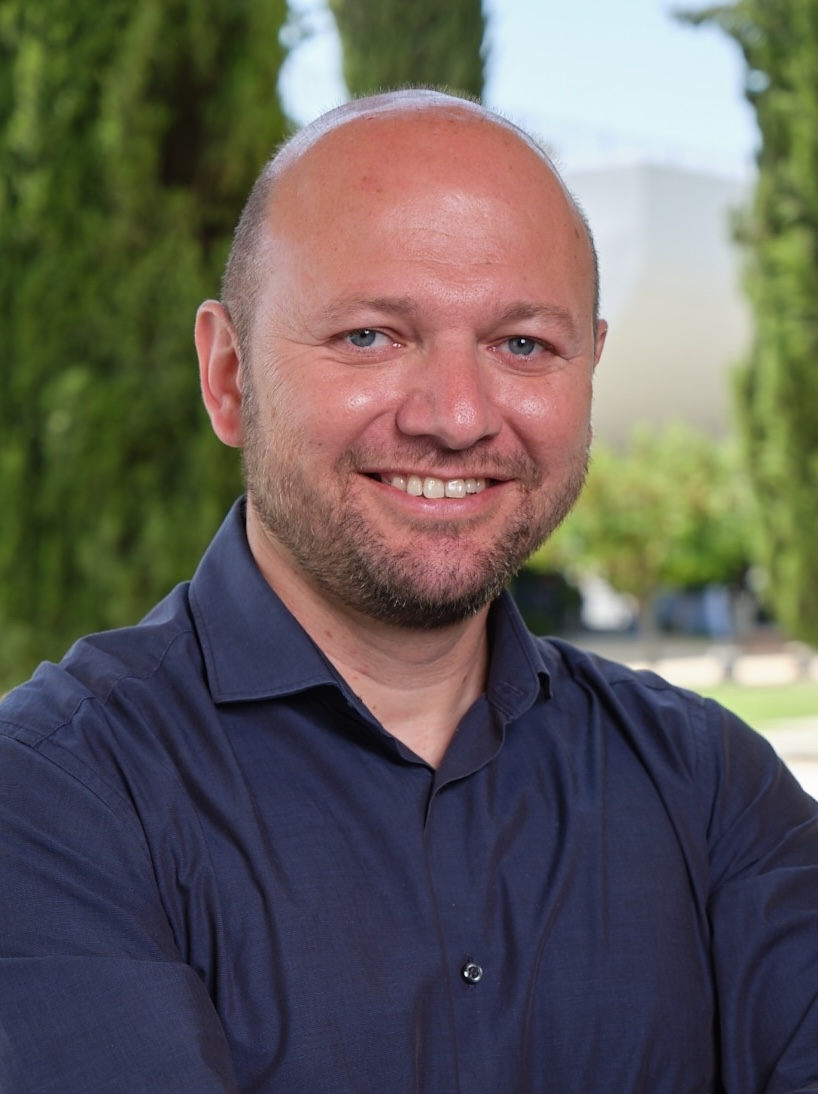}}]{Themistoklis Charalambous}~(Senior Member, IEEE) received the BA (First Class
Honours) and M.Eng. (Distinction) degrees in Electrical and Information Sciences from Trinity College, University of Cambridge, and the Ph.D. degree from the Control Laboratory, Department of Engineering, University of Cambridge.
Following his PhD, he held research and teaching positions at
Imperial College London, the University of Cyprus, KTH Royal Institute of Technology, and Chalmers University of Technology.
In 2017, he joined the Department of Electrical Engineering and Automation, School of Electrical Engineering, Aalto University, as an Assistant Professor; he was awarded the
title of Academy Research Fellow by the Academy of Finland in 2018 and was promoted to Associate Professor in 2020.
In September 2021, he joined the Department of Electrical and Computer Engineering, University of Cyprus, as a tenure-track Assistant Professor. 
%In May 2022, he received an ERC Consolidator Grant for the project ``MINERVA: Emerging Cooperative Autonomous Systems: Information for Control and Estimation.'' 
Since May 2026, he has been a tenured Associate Professor at the University of Cyprus. He is also a Visiting Professor at Aalto University and, since April 2023, at the FinEst Centre for Smart Cities.
%In September 2021, he continued his academic career at the University of Cyprus, where he joined the Department of Electrical and Computer Engineering as a tenure-track Assistant Professor. In May 2022, he was awarded an European Research Council (ERC) Consolidator Grant for the research project “MINERVA: Emerging Cooperative Autonomous Systems: Information for Control and Estimation.” From May 2026, he was appointed as a tenured Associate Professor at the University of Cyprus. He also remains affiliated with Aalto University as a Visiting Professor and, since April 2023, serves as a Visiting Professor at the FinEst Centre for Smart Cities.
%
%Prof. Charalambous is an Associate Editor of \textsc{IEEE Transactions on Automatic Control}
\end{IEEEbiography}

% ========================================
%
%
% That's all folks! :-)
%
%
% ========================================
\end{document}